\documentclass[runningheads]{llncs}
\usepackage[T1]{fontenc}
\usepackage{graphicx}
\usepackage{amsfonts}
\usepackage{amsmath}
\usepackage{stmaryrd}
\usepackage{xspace}
\usepackage{subcaption}
\usepackage{tikz}
\usetikzlibrary{
    arrows,
    decorations.pathmorphing,
    positioning,fit,trees,shapes,shadows,automata,calc,calligraphy
}
\usepackage{hyperref}
\usepackage{cleveref}
\usepackage{booktabs}
\usepackage{dirtytalk}
\usepackage{mathtools}
\usepackage[table,dvipsnames]{xcolor}
\usepackage{nicefrac}
\usepackage{float}

\usepackage{bbm}

\usepackage{thm-restate}

\usepackage[textsize=scriptsize,bordercolor=white,disable]{todonotes}
\setuptodonotes{fancyline}
\tikzset{
node distance=1.7cm,
state/.style={draw,circle,outer sep=0pt,minimum size=20pt,inner sep=1pt},
initial/.style={
    initial by arrow, initial text=, initial distance=4mm, draw
},
labeledinitial/.style={
  append after command={
        ([xshift=-4mm]\tikzlastnode.west)
        edge[edge] node[auto]{#1}
        (\tikzlastnode) 
    }
},
labeledinitial top/.style={
    append after command={
        ([yshift=4mm]\tikzlastnode.north)
        edge[edge] node[auto]{#1}
        (\tikzlastnode) 
    }
},
final/.style={
            append after command={
                (\tikzlastnode) edge[edge] ([xshift=4mm]\tikzlastnode.east)
            }
},
labeledfinal/.style={
    append after command={
        (\tikzlastnode) edge[edge] node[auto]{#1} ([xshift=4mm]\tikzlastnode.east)
    }
},
final bot/.style={
state,
            append after command={
                (\tikzlastnode) edge[edge] ([yshift=-4mm]\tikzlastnode.south)
            }
},
labeledfinal bot/.style={
    state,
    append after command={
        (\tikzlastnode) edge[edge] node[auto]{#1} ([yshift=-4mm]\tikzlastnode.south)
    }
},
inv final bot/.style={
state,
            append after command={
                (\tikzlastnode) edge[edge, white] ([yshift=-4mm]\tikzlastnode.south)
            }
},
phantom/.style={draw=none, fill=none},
edge/.style={->},
loop/.style={edge},
loop above/.style={
	edge,
	out=120,
	in=60,
	looseness=8
},
loop below/.style={
    edge,
    out=240,
    in=300,
    looseness=8
  },
  loop left/.style={
    edge,
    out=150,
    in=210,
    looseness=8
  },
  loop right/.style={
    edge,
    out=30,
    in=330,
    looseness=8
  },
underbrace/.style={
decorate,decoration={calligraphic brace,amplitude=10pt,mirror, raise=0.5cm}
},
subautomaton/.style={
draw, rectangle,
text width=1cm,
align=center,
text centered,
inner sep=0.15cm,
minimum height=0.25cm
},
}
\newcommand\sbullet[1][.5]{\mathbin{\vcenter{\hbox{\scalebox{#1}{$\bullet$}}}}}

\newcommand{\prog}{\ensuremath{\mathcal{P}}}

\newcommand{\concat}{\sbullet}

\newcommand{\V}{V} 

\newcommand{\set}[1]{\ensuremath{
\left\{#1\right\}
}}

\newcommand{\states}{Q}
\newcommand{\Raut}{\mathcal{A}}

\newcommand{\Baut}{\mathcal{B}}
\newcommand{\redip}{\texttt{ReDiP}\xspace}

\newcommand{\bfalse}{\texttt{false}}

\newcommand{\Aut}{\mathcal{A}}

\newcommand{\supp}[1]{\ensuremath{\mathrm{supp}(#1)}}

\newcommand{\trans}[2]{\ensuremath{\llbracket{#1}\rrbracket(#2)}}

\makeatletter
\providecommand*{\cupdot}{%
  \mathbin{%
    \mathpalette\@cupdot{}%
  }%
}
\newcommand*{\@cupdot}[2]{%
  \ooalign{%
    $\m@th#1\cup$\cr
    \sbox0{$#1\cup$}%
    \dimen@=\ht0 %
    \sbox0{$\m@th#1\cdot$}%
    \advance\dimen@ by -\ht0 %
    \dimen@=.5\dimen@
    \hidewidth\raise\dimen@\box0\hidewidth
  }%
}
\providecommand*{\bigcupdot}{%
  \mathop{%
    \vphantom{\bigcup}%
    \mathpalette\@bigcupdot{}%
  }%
}
\newcommand*{\@bigcupdot}[2]{%
  \ooalign{%
    $\m@th#1\bigcup$\cr
    \sbox0{$#1\bigcup$}%
    \dimen@=\ht0 %
    \advance\dimen@ by -\dp0 %
    \sbox0{\scalebox{2}{$\m@th#1\cdot$}}%
    \advance\dimen@ by -\ht0 %
    \dimen@=.5\dimen@
    \hidewidth\raise\dimen@\box0\hidewidth
  }%
}
\makeatother

\newcommand{\guard}{\varphi}

\newcommand{\geom}[2]{\normalfont{\mathtt{Geom}_{#2}\mathtt{(#1)}}\xspace}
\newcommand{\bern}[2]{\ensuremath{\mathtt{Bern}_{#2}\mathtt{(#1)}}\xspace}
\newcommand{\dirac}[2]{\ensuremath{\mathtt{Dirac}_{#2}\mathtt{(#1)}}\xspace}
\newcommand{\unif}[2]{\ensuremath{\mathtt{Unif}_{#2}\mathtt{(#1)}}\xspace}
\newcommand{\negbinomial}[3]{\texttt{NegBin}_{#3}\texttt(#1\texttt,#2\texttt)}

\newcommand{\semantics}[1]{\ensuremath{\llbracket #1 \rrbracket}}
\newcommand{\N}{\ensuremath{\mathbb{N}}}

\newcommand{\R}{\ensuremath{\mathbb{R}}}
\newcommand{\Rgez}{\ensuremath{\R_{\normalfont\texttt+}}}
\newcommand{\Rgezinf}{\ensuremath{{\Rgez^\infty}}}
\newcommand{\Bool}{\ensuremath{\mathbb{B}}}

\newcommand{\Y}{\ensuremath{\mathbf{Y}}}

\newcommand{\semiring}{\mathbb{S}}
\newcommand{\sringsubset}{{\semiring'}} 
\newcommand{\monoid}{\mathbb{M}}

\newcommand{\lin}{{\Rgez\V}}

\let\stmaryrdLightning\lightning

\newcommand{\skp}{\normalfont\texttt{skip}}

\newcommand{\assign}[2]{\normalfont#1\,\texttt{:=} ~ #2\xspace}

\newcommand{\incr}[2]{\normalfont#1\,\texttt{+=} ~ #2\xspace}

\newcommand{\decr}[1]{\normalfont #1\texttt{-}\texttt{-}\xspace}

\newcommand{\coinflip}[3]{\normalfont \{#1\} ~ \texttt[#2\texttt] ~ \{#3\}}

\newcommand{\iid}[2]{\normalfont \texttt{iid}\texttt(#1\texttt,#2\texttt)\xspace}

\newcommand{\ite}[3]{\normalfont \texttt{if} ~ \texttt(#1\texttt) ~ \texttt\{#2\texttt\} ~ \texttt{else} ~ \texttt\{#3\texttt\}\xspace}

\newcommand{\observe}[1]{\normalfont\texttt{observe} ~ \texttt( #1 \texttt)}

\newcommand{\fps}[2]{\ensuremath{#1\langle\!\langle#2\rangle\!\rangle}}

\newcommand{\weightedunion}[2]{
\ensuremath{
#1 \oplus #2
}
}

\newcommand{\probchoice}[4]{
#1 \prescript{#3}{}{\oplus}^{#4} #2
}

\newcommand{\semanticsaut}[1]{|\!|{#1}|\!|}

\newcommand{\ruletag}[1]{\normalfont(\textsc{#1})}
\newcommand{\sosrule}[3]{\ruletag{#1}\,\frac{#2}{#3}}
\newcommand{\opstate}[2]{\langle\,#1\,,\,#2\,\rangle} 
\newcommand{\val}{\sigma} 
\newcommand{\done}{\downarrow}

\newcommand{\markovchain}[2]{\ensuremath{
\mathcal{M}_{#1}\semantics{#2}
}}

\newcommand{\markovstates}{\ensuremath{S}}
\newcommand{\markovtrans}{\ensuremath{P}}
\newcommand{\markovinit}{\ensuremath{I}}

\newcommand{\markovprob}[2]{\ensuremath{\mathrm{Pr}^{\markovchain{\Raut}{#1}}(#2)
}
}

\newcommand{\normalize}[1]{\ensuremath{\mathrm{norm}(#1)}}

\makeatletter
\newcommand{\monus}{\mathbin{\text{\@dotminus}}}

\newcommand{\@dotminus}{%
  \ooalign{\hidewidth\raise1ex\hbox{.}\hidewidth\cr$\m@th-$\cr}%
}
\makeatother

\newcommand{\poisson}[1]{\texttt{Poisson}\texttt(#1\texttt)}
\newcommand{\gammadist}[2]{\texttt{Gamma}\texttt(#1\texttt,#2\texttt)}
\newcommand{\negbinom}[2]{\texttt{NegBin}\texttt(#1\texttt,#2\texttt)}
\newcommand{\binomial}[2]{\texttt{Binom}\texttt(#1\texttt,#2\texttt)}

\newcommand{\inlineTrans}[1]{\bigcirc\!\!\xrightarrow{#1}\!\!\bigcirc}
\newcommand{\inlineTransTwo}[2]{\bigcirc\!\!\xrightarrow{#1}\!\!\bigcirc\!\!\xrightarrow{#2}\!\!\bigcirc}
\newcommand{\labelSubs}[3]{#1[#2/#3]}
\newcommand{\myred}{red!80!black}
\newcommand{\red}[1]{\textcolor{\myred}{#1}}
\newcommand{\gray}[1]{\textcolor{gray}{#1}}

\newcommand{\sizeaut}[1]{\ensuremath{|#1|}}
\newcommand{\sizebool}[1]{\ensuremath{|#1|}}
\newcommand{\sizeprog}[1]{\ensuremath{|#1|}}
\newcommand{\bigO}[1]{\ensuremath{\mathcal{O}(#1)}}

\newcommand{\probmass}[1]{\ensuremath{\sum_{\sigma\in\N^V} \semanticsaut{#1}(\sigma)}}

\newcommand{\errorstate}{\ensuremath{\langle\stmaryrdLightning\rangle}}
\newcommand{\semanticspath}[1]{\ensuremath{|\!|#1|\!|}}
\newcommand{\semanticsapath}[1]{\overline{\semanticspath{#1}}}
\newcommand{\APaths}[1]{\ensuremath{\textit{APaths}(#1)}}
\newcommand{\lengthpath}[1]{\ensuremath{|#1|}}

\usepackage{color}

\begin{document}
\title{Weighted Automata for Exact Inference in Discrete Probabilistic Programs}
%
%
\author{Dominik Geißler\inst{1}\orcidID{0009-0008-8069-1417} \and
Tobias Winkler\inst{2}\orcidID{0000-0003-1084-6408}}
\authorrunning{D. Geißler and T. Winkler}
%
\institute{Technische Universität Berlin, Berlin, Germany\and
RWTH Aachen University, Aachen, Germany}
\maketitle              
\begin{abstract}
    In probabilistic programming, the \emph{inference problem} asks to determine a program's posterior distribution conditioned on its ``observe'' instructions.
    Inference is challenging, especially when exact rather than approximate results are required.
    Inspired by recent work on probability generating functions (PGFs), we propose encoding distributions on $\N^k$ as weighted automata over a commutative alphabet with $k$ symbols.
    Based on this, we map the semantics of various imperative programming statements to automata-theoretic constructions.
    For a rich class of programs, this results in an effective translation from prior to posterior distribution, both encoded as automata.
    We prove that our approach is sound with respect to a standard operational program semantics.
\end{abstract}

\setcounter{footnote}{0} 

\section{Introduction}
\label{sec:intro}
\emph{Probabilistic programming languages} extend traditional programming languages by
capabilities for \emph{sampling} numbers from pre-defined distributions, and
\emph{conditioning} the current program state on observations~\cite{probprogramming}.
Probabilistic programs have numerous applications, including machine learning~\cite{deepprohprogramming}, cognitive science~\cite{cogscience}, and autonomous systems~\cite{autsystems}.
Semantically, probabilistic programs can be seen as \emph{transformers of probability distributions}:
From an initial distribution over inputs, also called \emph{prior}, to a final distribution over outputs, also called \emph{posterior}~\cite{DBLP:journals/jcss/Kozen81,probprogramming}.
\emph{Inference} means characterizing the posterior resulting from a given prior.

In this paper, we study the imperative language \emph{\redip}~\cite{redip,credip}, short for \emph{rectangular discrete probabilistic programming language}.
This language imposes some syntactical restrictions (see \Cref{sec:programs} for details), while preserving decidability of many (inference) tasks:
For instance, for loop-free programs, the moments of the posterior distributions can be computed exactly, and program equivalence is decidable~\cite{redip}.
We only consider the \emph{loop-free} fragment of \redip in this paper; adding general \texttt{while} loops renders the language Turing-complete.
A key challenge related to \redip is its support for some \emph{infinite-support} distributions.
Implementing the distribution transformer semantics directly thus requires manipulating such distributions in an effective manner.
In \cite{generatingfunctionsfor,redip,credip}, symbolic closed-form expressions of \emph{probability generating functions} (PGFs) were employed for this purpose.
In this paper, we propose an automata-theoretic alternative to the PGF approach.
Our main idea is to encode distributions over the $\N$-valued variables of a \redip program by means of \emph{weighted automata} over a commutative alphabet.
This automaton model, which we call \emph{probability generating automata} (PGA), is a restricted form of a weighted multi-counter system:
Intuitively, each time a PGA takes an $X$-transition, the program variable represented by $X$ is incremented.
The probability that $X=n$ in the distribution described by a PGA is thus the sum of the weights of all paths containing exactly $n$ many $X$-transitions; this extends naturally to joint distributions. 

\begin{figure}[t]
    \begin{minipage}{0.45\textwidth}
        \centering
        \begin{align*}
            & \gray{\texttt{// all variables initially 0}} \\
            & \coinflip{\assign{R}{0}}{\nicefrac{9}{10}}{\assign{R}{1}} \,\fatsemi \\
            & \mathtt{if} ~ (R  = 0) ~ \{ \\
            & \qquad \incr{X}{\negbinom{1}{\nicefrac 1 2}} \\
            & \} ~ \texttt{else} ~ \{ \\
            & \qquad \incr{X}{\negbinom{2}{\nicefrac 1 2}} \\
            & \} \,\fatsemi \\
            & \observe{X \geq 2}
        \end{align*}
    \end{minipage}
    \begin{minipage}{0.54\textwidth}
        \centering
        \begin{tikzpicture}[node distance=5mm and 8mm]
            \node[labeledinitial={$\frac{40}{11}$}, state] (i) {};
            \node[state,above right=of i,yshift=-3mm] (1) {};
            \node[state,right=of 1] (2) {};
            \node[labeledfinal={$\tfrac{1}{2}$}, state, right =of 2] (3) {};
            
            \node[state, below right =of i,yshift=3mm] (l1) {};
            \node[state, right =of l1] (l2) {};
            \node[state, right =of l2] (l3) {};
            \node[state, below =of l1] (l4) {};
            \node[state, right =of l4] (l5) {};
            \node[labeledfinal={$\tfrac{1}{2}$},state, below =of l3] (l6) {};
            
            \draw[edge] (i) -- node[auto]{$\tfrac{9}{10}$} (1);
            \draw[edge] (1) -- node[auto]{$\tfrac{1}{2}X$} (2);
            \draw[edge] (2) -- node[auto]{$\tfrac{1}{2}X$} (3);
            
            \draw[edge] (i) -- node[below left]{$\tfrac{1}{10}R$} (l1);
            \draw[edge] (l1) -- node[auto]{$\tfrac{1}{2}X$} (l2);
            \draw[edge] (l2) -- node[auto]{$\tfrac{1}{2}X$} (l3);
            \draw[edge] (l3) -- node[auto]{$\tfrac{1}{2}$} (l6);
            \draw[edge] (l1) -- node[auto]{$\tfrac{1}{2}$} (l4);
            \draw[edge] (l2) -- node[auto]{$\tfrac{1}{2}$} (l5);
            \draw[edge] (l4) -- node[auto]{$\tfrac{1}{2}X$} (l5);
            \draw[edge] (l5) -- node[auto]{$\tfrac{1}{2}X$} (l6);
            
            \draw[loop above,looseness=5] (3) to node[right=2mm,yshift=1mm]{$\tfrac{1}{2}X$} (3);
            \draw[loop above,looseness=5] (l3) to node[right=2mm,yshift=-2mm]{$\tfrac{1}{2}X$} (l3);
            \draw[loop below,looseness=5] (l6) to node[right=2mm,yshift=2mm]{$\tfrac{1}{2}X$} (l6);
        \end{tikzpicture}
    \end{minipage}
    \caption{
        \emph{Left:} High-risk policyholder model in \redip.
        \emph{Right:} Illustration of the PGA construction for the posterior distribution (defined in detail in \Cref{sec:semantics}). $\tfrac{40}{11}$ is the normalization factor resulting from the \texttt{observe}-operation as described in \Cref{sec:observe}.
    }
    \label{fig:motivatingExample}
\end{figure}

\paragraph{Motivating Example: Inferring High-Risk Policyholders.}
Suppose that an insurance company models the number $X$ of claims filed by each individual policyholder as a Poisson-Gamma model $\theta \sim \gammadist{\alpha}{\beta}; X \sim \poisson{\theta}$, which is equivalent to a negative binomial model $X \sim \negbinom{\alpha}{\frac{\beta}{1+\beta}}$~\cite{Willmot_1986}.\footnote{Poisson and gamma distributions are not supported directly by our language.}
For simplicity, we further assume a Bernoulli prior on $\alpha$ ($\alpha = 1$ with prior probability $0.9$; $\alpha=2$ with prior probability $0.1$) and fix $\beta = 1$; these choices imply that low-risk customers ($\alpha = 1$) file one claim on average, whereas high-risk policyholders ($\alpha = 2$) file two claims.
Now consider the following scenario: An employee of the insurance notices that she is processing, for the second time, a claim of the same policyholder---hence she knows that this individual has filed \emph{at least} two claims.
A typical inference problem is to determine the probability of the customer being in the high-risk category.
This situation can be modeled in \redip as shown in \Cref{fig:motivatingExample}, and solved\footnote{The posterior probability of being high-risk is exactly $\nicefrac{2}{11} \approx 0.182$.} by analyzing the resulting automaton (see \Cref{ex:semanticsMotivating}).

\paragraph{Contributions.}
In summary, the contributions of this paper are as follows:
\begin{itemize}
    \item We propose encoding joint distributions---possibly with infinite support---over non-negative integer variables as certain weighted automata (\Cref{sec:wa}).
    \item We provide an effective translation from programs to automata transformations, enabling posterior inference via common automata-theoretic constructions (\Cref{sec:programs,sec:semantics}).
    \item We prove our approach sound with regards to an existing \emph{operational} program semantics in terms of Markov chains~\cite{DBLP:journals/toplas/OlmedoGJKKM18} (\Cref{sec:operational}).
\end{itemize}
\section{Background on Weighted Automata}
\label{sec:wa}
We define $\Rgez = \set{r \in\R \mid r \geq 0}$, $\Rgezinf = \Rgez \cup \set{\infty}$, and $\Bool = \{0,1\}$.

\paragraph{Semirings and Formal Power Series.}
The following definitions closely follow~\cite{DBLP:reference/hfl/Kuich97}.
A \emph{semiring} is a 5-tuple $(\semiring, +, \cdot, 0, 1)$ where $(\semiring, +, 0)$ is a commutative monoid, $(\semiring, \cdot, 1)$ is a monoid, $\cdot$ distributes over $+$, and $a \cdot 0 = 0\cdot a = 0$ for all $a \in\semiring$.
If the operations and neutral elements are clear from the context, we only use $\semiring$ to refer to the semiring.
We call a semiring $\semiring$ \emph{naturally ordered} if the binary relation $a \sqsubseteq b \overset{\text{\tiny def}}{\iff} \exists c\in \semiring \colon a + c = b$ forms a partial order on $\semiring$.
A semiring together with a sum operator $\sum_{i \in I} a_i$ defined for arbitrary---possibly infinite---families $(a_i)_{i \in I}$ is called \emph{complete} if $\sum_{i \in I} a_i$ behaves as usual finite sums for finite $I$, and is commutative as well as distributive (see~\cite{DBLP:reference/hfl/Kuich97} for the formal definitions).
    A naturally ordered, complete semiring $\semiring$ is called \emph{$\omega$-continuous} if the following condition holds for all \emph{countable} families $(a_i)_{i\in\N}$:
    $
        \forall c \in \semiring \colon \forall n\in \N \colon 
        \sum_{i \leq n} a_i \sqsubseteq c
        \implies
        \sum_{i\in\N}a_i \sqsubseteq c
        .
    $

Two $\omega$-continuous semirings of interest are the \emph{Boolean semiring}
$(\Bool, \lor, \land, 0, 1)$ and the \emph{non-negative extended real semiring} $(\Rgezinf, +,\cdot, 0, 1)$.
In the latter, the sum operator is defined via $\sum_{i\in I} a_i = \sup \set{\sum_{i \in E} a_i \mid E\subseteq I, E \text{ finite}}$, which possibly evaluates to the $\infty$ element.

Given an $\omega$-continuous semiring $(\semiring, +, \cdot, 0, 1)$ and a monoid $(\monoid, \circ, e)$, we define the set $\fps{\semiring}{\monoid}$ of \emph{formal power series} (FPS) as the set of maps $\{f \colon \monoid \to \semiring\}$.
As usual, FPS are denoted as ``formal infinite sums''
$
    f = \sum_{m \in \monoid} f(m) m
$.
We call $f(m)$ the \emph{coefficient} of $m \in \monoid$ and define the \emph{support} of $f$ as $\supp{f} = \{m \in \monoid \mid f(m) \neq 0\}$.
$\fps{\semiring}{\monoid}$ is a semiring with the following operations:
For all $f_1,f_2 \in \fps{\semiring}{\monoid}$ and $m \in \monoid$, $(f_1 + f_2)(m) = f_1(m) + f_2(m)$ and $(f_1\cdot f_2)(m) = \sum_{m=m_1 \circ m_2} f_1(m_1) \cdot f_2(m_2)$.
The neutral elements are $0$ and $1e = e$. If $\semiring$ is $\omega$-continuous, then so is $\fps{\semiring}{\monoid}$~\cite{DBLP:reference/hfl/Kuich97}.

\begin{example}[Formal Power Series]
    \label{ex:fps}
    We illustrate FPS by means of the instances relevant to this paper.
    The overall idea is that FPS generalize the concept of (weighted) formal languages. 
    Let $\V$ be a finite alphabet.
    \begin{itemize}
        \item $\V^*$ is a monoid with concatenation and neutral element $\varepsilon$, and $\fps{\Bool}{\V^*}$ is isomorphic to the \emph{semiring of formal languages} $2^{\V^*}$ over $\V$.
        In this semiring, addition is union and multiplication is word-wise concatenation of languages.
        \item Languages over \emph{commutative} symbols can be modeled using the monoid $\N^\V = \{\val \colon \V \to \N\} \cong \N^{|\V|}$ with pointwise addition.
        In this paper, we employ FPS from $\fps{\Rgezinf}{\N^\V}$ to encode distributions over $\N^\V$.
        For instance,
        \[
            \tfrac 1 2 Y ~+~ \tfrac 1 4 X Y^2 ~+~ \tfrac 1 8 X^2 Y^3 ~+~ \ldots \quad \in~ \fps{\Rgezinf}{\N^\V}
        \]
        describes a joint distribution over the variables $V = \{X, Y, \ldots\}$ where $X$ is geometrically distributed and, with probability $1$, $Y = X+1$ and $Z = 0$ for all $Z \in \V \setminus \{X,Y\}$.
        Note that we use \emph{monomial notation}\footnote{Formally, assuming $\V = \{X_1,\ldots,X_k\}$, $\sigma \in \N^\V$ is written as $X_1^{\sigma(X_1)} \ldots X_k^{\sigma(X_k)}$ where variables with an exponent of $0$ are omitted.} for the variable valuations $\N^\V$.
    \end{itemize}
\end{example}

\paragraph{Weighted Automata over Semirings.}
We adopt the following automaton model:
\begin{definition}[Weighted Automaton~\cite{DBLP:reference/hfl/Kuich97}]
    Let $\semiring$ be an $\omega$-continuous semiring and let $\emptyset\neq \sringsubset \subseteq \semiring$.
    An \emph{$\sringsubset$-automaton over $\semiring$} is a 4-tuple $\Aut = (\states,M,I,F)$ where
    $\states \neq \emptyset$ is a finite set of \emph{states},
    $M \in \sringsubset^{\states\times \states}$ is a \emph{transition matrix}, and
    $I \in \sringsubset^{1\times \states}$ and $F \in \sringsubset^{\states \times 1}$ are vectors of initial and \emph{final weights}, respectively.
\end{definition}
Notice that only the elements in $\sringsubset$ are allowed as transition labels and initial/final weights. 
A state $s \in \states$ is called \emph{initial} \emph{(final)} if $I_s \neq 0$ ($F_s \neq 0$), where $0$ refers to the semiring zero.
The \emph{behavior} (or semantics) of $\Aut$ is defined as
\[
    \semanticsaut{\Aut} ~=~ IM^*F ~\in~ \semiring ~,
    \qquad \text{where} \quad
    M^* = \sum_{n\in \N} M^n
    ~.
\]
The infinite sum $M^*$ is well-defined because $\semiring$ is complete by assumption.
Intuitively, $\semanticsaut{\Aut}$ is the semiring-sum of the weights of all finite-length paths in $\Aut$, where the weight of a path is the semiring-product of the initial weight of its first state, the weights of its transitions, and the weight of the final state.
As usual, we often represent automata graphically; we use \emph{unlabeled} dangling arrows for initial and final states with weight $1$, and \emph{labeled} dangling arrows for initial and final states with a weight other than $1$.

We only consider automata over FPS semirings $\fps{\semiring}{\monoid}$ in this paper.
Let $\Aut = (\states,M,I,F)$ be such an automaton.
We write $\Aut \colon s \xrightarrow[]{am} t$ iff $M_{s,t}(m) = a$, where $s,t, \in \states$, $a \in \semiring$, $m \in \monoid$; $\Aut$ may be omitted if clear from context.
Further, we say that $\Aut$ is \emph{normalized}\footnote{The notion of normalization in \cite{DBLP:reference/hfl/Kuich97} is slightly stronger than ours.} if $\bigcup_{s \in \states}\supp{I_s} \cup \supp{F_s}$ is either $\emptyset$ or $\{e\}$, where $e$ is the unit of $\monoid$.
For $\fps{\Bool}{\V^*}$ and $\fps{\Rgezinf}{\N^\V}$ (see \Cref{ex:fps}), this condition means that the initial and final weights are elements from $\Bool$ and $\Rgezinf$, respectively.

\begin{example}[Relevant Automata in this Paper]
    Let $\V$ be a finite alphabet.
    \begin{itemize}
        \item Define the following subset of $\fps{\Rgezinf}{\N^\V}$:
        \[
            \lin ~=~ \{rX \mid r \in \Rgez, X \in \V\} \cup \Rgez
        \]
        The protagonists of our method are $\lin$-automata over the semiring $\fps{\Rgezinf}{\N^\V}$ (see \Cref{fig:motivatingExample} for an example).
        The behavior of such an automaton is an FPS where the coefficient of $\val \in \N^\V$ equals the sum of the weights of all paths in which every $X \in \V$  appears exactly $\val(X)$ many times.
        An $\lin$-automaton is thus similar to a weighted multi-counter system with $|\V|$ counters admitting increment operations only.
        \item A \emph{nondeterministic finite automaton} (NFA) over alphabet $\V$ is a normalized $2^\V$-automaton over the \emph{semiring of formal languages} $\fps{\Bool}{\V^*} \cong 2^{\V^*}$.
        A \emph{deterministic} finite automaton (DFA) is an NFA with the additional restrictions that (i) there exists exactly one initial state, and (ii) for all $s \in \states$ and $X \in \V$ there exists exactly one $t \in \states$ such that $s \xrightarrow{X} t$.
    \end{itemize}
\end{example}

\paragraph{Probability Generating Functions and Automata.}
A \emph{probability generating function} (PGF) is an FPS $f \in \fps{\Rgez}{\N^\V}$ with \emph{mass} $\sum_{\val \in \N^\V} f(\val) \leq 1$.
A PGF thus describes a probability (sub-)distribution on $\N^\V$, see \Cref{ex:fps}.

\begin{definition}[PGA: Probability Generating Automaton]%
    \label{def:pga}%
    A \emph{PGA} is an $\lin$-automaton $\Raut$ over $\fps{\Rgezinf}{\N^\V}$ s.t.\ $\semanticsaut{\Raut}$ is a PGF, i.e.\ $\sum_{\val \in \N^\V} \semanticsaut{\Raut}(\val) \leq 1$.
\end{definition}
Although all coefficients are included in $[0,1]$, $[0,1]$ alone is not sufficient as we need elements from $\R_{>1}$ for normalization (see \Cref{fig:motivatingExample}). 

\Cref{fig:basicDistrAutomata} provides examples of PGA for basic distributions;
it can be checked that $\semanticsaut{\Raut_{\geom{p}{X}}} = \sum_{i\geq 0} (1{-}p)^i p X^i$, $\semanticsaut{\Raut_{\bern{p}{X}}} = 1{-}p + pX$, $\semanticsaut{\Raut_{\dirac{n}{X}}} = X^n$, and $\semanticsaut{\Raut_{\unif{m}{X}}} = \sum_{i=0}^{m-1} \tfrac 1 m X^i$.
More generally, we say that a distribution is \emph{PGA-definable} if its PGF is the behavior of some PGA.
It follows from the definitions that the class of PGA-definable distributions includes the so-called \emph{discrete-phase distributions}~(see, e.g.,~\cite{navarro2019order}); a more thorough characterization of PGA-definable distributions is left for future work.

It may not always be obvious if a given automaton is a PGA, i.e.\ if the sum of the weights of its finite paths is at most $1$.
However, the following lemma asserts that this can be tested in polynomial time:

\begin{lemma}
    \label{thm:computeMass}
    For every given\footnote{This statement assumes binary-encoded \emph{rational} numbers as transition weights.} $\lin$-automaton $\Raut$ over $\fps{\Rgezinf}{\N^\V}$, the exact mass $\sum_{\val \in \N^\V} \semanticsaut{\Raut}(\val) \in \Rgezinf$ can be computed in polynomial time in the size of $\Raut$.
\end{lemma}
\begin{proof}
    We construct an $\Rgez$-automaton $\Raut' = (Q, M, I, F)$ from $\Raut$ by removing the symbols $\V$ from the transitions of $\Raut$ (see ``transition label substitution'' in \Cref{sec:setToZero}), as well as from its initial and final weights.
    It can be verified that $\semanticsaut{\Raut'} = \sum_{\val \in \N^\V} \semanticsaut{\Raut}(\val)$.
    Next, we use that $\semanticsaut{\Raut'} = IB$, where $B$ is the componentwise least solution of the linear equation system $B = MB + F$ over \Rgezinf, see~\cite[Thm.~4.1]{DBLP:reference/hfl/Kuich97}.
    Consequently, $\semanticsaut{\Raut'}$ is either the optimal value of the linear program
    “minimize $IB$ s.t.\ $B = MB + F \land B \geq 0$,”
    or $\infty$ if the linear program is infeasible.
    In either case, the outcome can be determined in polynomial time.
\end{proof}

\begin{figure}[t]
    \centering
    \setlength{\tabcolsep}{10pt}
    \begin{tabular}{c c}
        \begin{tikzpicture}[node distance=2mm and 7mm]
            \node[state,labeledfinal=$p$,initial] (q0) {};
            \draw[loop above,looseness=5] (q0) to node[right,xshift=2mm] {$(1{-}p)X$} (q0);
        \end{tikzpicture}
        &
        \begin{tikzpicture}[node distance=2mm and 5mm]
            \node[initial, state] (q0) {$0$};
            \node[state, right=of q0] (q1) {$1$};
            \node[right= of q1] (ph) {$\hdots$};
            \node[final,state, right= of ph] (qn) {$n$};
            \draw[edge] (q0) -- node[above]{$X$} (q1);
            \draw[edge] (q1) -- node[above]{$X$} (ph);
            \draw[edge] (ph) -- node[above]{$X$} (qn);
        \end{tikzpicture} \\
        $\Raut_{\geom{p}{X}}$ & $\Raut_{\dirac{n}{X}}$ \\[1em]
        \begin{tikzpicture}[node distance=2mm and 7mm]
            \node[state, labeledfinal bot=$1{-}p$,initial] (q0) {};
            \node[state, final bot, right=of q0](q1) {};
            \draw[edge] (q0) to node[above] {$p X $} (q1); 
        \end{tikzpicture}
        &
        \begin{tikzpicture}[node distance=2mm and 5mm]
            \node[labeledfinal bot=$\tfrac{1}{m}$,initial, state] (q0) {$0$};
            \node[labeledfinal bot=$\tfrac{1}{m}$,state, right=of q0] (q1) {$1$};
            \node[right= of q1] (ph) {$\hdots$};
            \node[labeledfinal bot=$\tfrac{1}{m}$,state, right= of ph] (qn) {\scriptsize$m{-}1$};
            \draw[edge] (q0) -- node[above]{$X$} (q1);
            \draw[edge] (q1) -- node[above]{$X$} (ph);
            \draw[edge] (ph) -- node[above]{$X$} (qn);
        \end{tikzpicture} \\
        $\Raut_{\bern{p}{X}}$ & $\Raut_{\unif{m}{X}}$ \\
    \end{tabular}
    \caption{PGA for basic distributions ($X \in V, p \in [0,1], n \in \N, m \in \N_{>0}$ are constants).}
    \label{fig:basicDistrAutomata}
\end{figure}
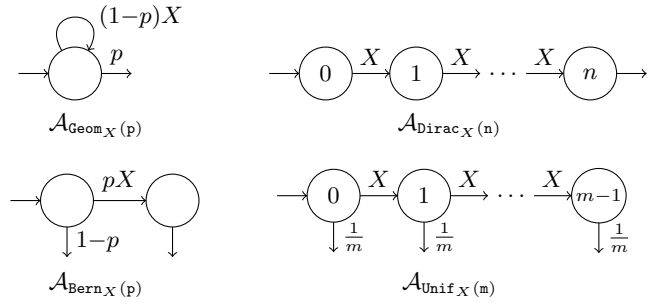
\section{Probabilistic Programs}
\label{sec:programs}

Let $\V$ be a finite alphabet of \emph{program variables}, fixed throughout the rest of the paper. The following is a slight extension\footnote{Specifically, we allow $\equiv_m$ (congruence modulo $m$) in guards, and explicitly incorporate some statements that are only available as syntactic sugar in~\cite{redip,credip}.} of (loop-free) \redip from~\cite{redip,credip}:

\begin{definition}[Syntax of \redip]%
    \label{def:syntax-redip}%
    \redip programs $\prog$ adhere to the following grammar (where $X, Y \in \V$, and $D$ is a PGA-definable distribution):
    \begin{align*}
        \prog \quad\Coloneqq\quad
                  & \assign{X}{0} \quad\quad \tag*{Set variable $X$ to $0$} \\         
        \mid\quad & \incr{X}{n} \tag*{Increment by constant $n \in \N$} \\
        \mid\quad & \incr{X}{D} \tag*{Increment by random sample from $D$} \\
        \mid\quad & \incr{X}{Y} \tag*{Increment by variable} \\
        \mid\quad & \incr{X}{\iid{D}{Y}} \tag*{Increment by sum of $Y$ i.i.d.\ samples from $D$} \\
        \mid\quad & \coinflip{\prog}{p}{\prog} \tag*{Random branching ($p \in [0,1]$)} \\
        \mid\quad & \ite{\guard}{\prog}{\prog} \tag*{Conditional branching} \\
        \mid\quad & \decr{X} \tag*{Decrement (``monus'' semantics)} \\
        \mid\quad & \observe{\guard} \tag*{Conditioning} \\
        \mid\quad & \prog\fatsemi \prog \tag*{Sequential composition} \\
        \guard \quad\Coloneqq\quad & X < n ~\mid~ X \equiv_{m} n ~\mid~ \guard \land \guard ~\mid~ \neg \guard  \tag*{Guards ($n, m \in\N, m > n$)}
    \end{align*}
\end{definition}
The particular instruction set of \redip is chosen so that each instruction roughly corresponds to an elementary automata-theoretic construction, see \Cref{sec:semantics}.
The intended effect of each statement should be sufficiently clear, perhaps with the exception of $\incr{X}{\iid{D}{Y}}$ and $\observe{\guard}$, which we explain in \Cref{sec:iid,sec:observe}, respectively.
Notice that the instructions 1--2 and 4 together allow assigning linear expressions with coefficients in $\N$ to variables.
For instance, $\assign{X}{1 + X + 2Y}$ can be expressed as $\assign{X}{0}\fatsemi \incr{X}{1} \fatsemi \incr{X}{Y} \fatsemi \incr{X}{Y}$.
We use $\skp$ as a shorthand for the effectless program $\incr{X}{0}$.

An important syntactic restriction of \redip is that guards can only compare variables to \emph{constants}, but not to other variables---hence the name ``rectangular''.
Other common comparison operators (e.g., $\geq$, $=$) and Boolean connectives (e.g., $\lor, \rightarrow$) are included as syntactic sugar.
The semantics $\semantics{\guard}$ of guards is standard and defined in \Cref{tab:guards}.
For $\val \in \N^\V$ we often write $\val \models \guard$ instead of $\val \in \semantics{\guard}$.

As a first step towards our automata-theoretic interpretation of \redip in the upcoming \Cref{sec:semantics}, we encode the rectangular guards $\guard$ from \Cref{def:syntax-redip} as DFA.
For $w \in \V^*$ we define its \emph{Parikh image} $\Psi(w) = \lambda X. |w|_X \in \N^\V$; that is, for every $X \in \V$, $\Psi(w)(X)$ is the number of occurrences of $X$ in $w$.
\begin{figure}[t]
    \centering
    \def\looseness{5}
\begin{tikzpicture}
\node[final bot, state, initial] at (0,0) (0) {$0$};
\node[phantom] at (1.5,0) (ph1) {$\hdots$};
\node[state, final] at (3,0) (n1) {$n{-}1$};
\node[state] at (3, -1.5) (n) {$n$};

\draw[loop above,looseness=\looseness] (0) to node[above]{$\Y$} (0);
\draw[edge] (0) -- node[above]{$X$} (ph1);
\draw[edge] (ph1) -- node[above]{$X$} (n1);
\draw[loop above,looseness=\looseness] (n1) to node[above]{$\Y$} (n1);
\draw[edge] (n1) -- node[right]{$X$} (n);
\draw[loop right,looseness=\looseness] (n) to node[right]{$X, \Y$} (n);


\node[initial, state] at (5,0) (0) {$0$};
\draw[loop above, looseness=\looseness] (0) to node[above]{$X,\Y$} (0);


\node[initial, state] at (7,0) (0) {$0$};
\node[phantom] at (8.5,0) (ph1) {$\hdots$};
\node[final, state] at (10,0) (n) {$n$};

\node[phantom] at (10, -1.5) (ph2) {$\hdots$};
\node[state] at (8.5, -1.5) (m) {$m {-} 1$};

\draw[edge] (0) -- node[above]{$X$} (ph1);
\draw[edge] (ph1) -- node[above]{$X$} (n);
\draw[edge] (n) -- node[right]{$X$} (ph2);
\draw[edge] (ph2) -- node[above]{$X$} (m);
\draw[edge] (m) -- node[auto]{$X$} (0);

\draw[loop above, looseness=\looseness] (0) to node[above]{$\Y$} (0);
\draw[loop above, looseness=\looseness] (n) to node[above]{$\Y$} (n);
\draw[loop above, looseness=\looseness] (m) to node[above]{$\Y$} (m);
\end{tikzpicture}
    \caption{DFA for guards $\Baut_{X<n}$ ($n>0$), $\Baut_{X<0}$ ($=\Baut_\bfalse$) and $\Baut_{X\equiv_m n}$ ($m>n$) (with $\Y = \V \setminus \set{X}$).}
    \label{fig:guardAutomata}
\end{figure}

\begin{table}[t]
    \centering
    \caption{Guard semantics and translation to automata.}
    \label{tab:guards}
    \setlength{\tabcolsep}{6pt}
    \rowcolors{2}{gray!7}{white} 
    \begin{tabular}{l l l}
        \toprule
        Guard $\guard$ & Semantics $\semantics{\guard}$ & Guard DFA $\Baut_\guard$ \\ \midrule
        $X < n$ & $\{\val \in \N^\V \mid \val(X) < n\}$ & see \Cref{fig:guardAutomata} \\
        $X \equiv_m n$ & $\{\val \in \N^\V \mid \val(X) = n \mod m\}$ & see \Cref{fig:guardAutomata} \\
        $\guard \land \guard'$ & $\semantics{\guard} \cap \semantics{\guard'}$ & $\Baut_{\guard} \times \Baut_{\guard'}$ (standard DFA product) \\
        $\neg\guard$ & $\N^\V \setminus \semantics{\guard}$ & $\overline{\Baut_\guard}$ (standard DFA complement) \\
        \bottomrule
    \end{tabular}
\end{table}

\begin{definition}[Automaton Modeling a Guard]%
    \label{def:modelingGuard}%
    We say that an NFA $\Baut$ over alphabet $\V$ \emph{models a guard $\guard$} if
    \[
        \forall w \in \V^* \colon\qquad \Baut \text{ accepts } w \quad\text{if and only if}\quad \Psi(w) \models \guard
        ~.
    \]
\end{definition}
It follows that the language accepted by an NFA modeling some guard is \emph{closed under permutations}.

\begin{lemma}[Properties of the Automata $\Baut_\guard$]%
    For every guard $\guard$, the automaton $\Baut_\guard$ (defined inductively in \Cref{tab:guards}) is a DFA which models $\guard$ in the sense of \Cref{def:modelingGuard}.
\end{lemma}
\begin{proof}
    The base automata in \Cref{fig:guardAutomata} are deterministic and model the corresponding atomic guards.
    The statement then follows by induction on the structure of $\guard$, using the standard constructions for DFA intersection and complement.
\end{proof}

\section{Interpreting \redip with Weighted Automata}
\label{sec:semantics}
As explained in \Cref{sec:intro}, we view \redip programs semantically as \emph{transformers of probability distributions}: from initial (prior) distributions to final (posterior) distributions.
Our philosophy is to represent such distributions with PGA.

In this section, we implement the distribution transformer semantics of \redip by means of effective automata-theoretic constructions.
More precisely, given a PGA $\Raut$ (fixed throughout the rest of this section) encoding an initial distribution and a \redip program $\prog$, we outline an algorithm to construct a PGA $\semantics{\prog}(\Raut)$ for the resulting (unnormalized\footnote{$\semantics{\prog}(\Raut)$ may describe a \emph{sub}-distribution even if $\semanticsaut\Raut$ is a proper distribution, i.e.\ has mass one. \Cref{sec:observe} describes how to obtain a PGA for the \emph{normalized} posterior.}) posterior distribution.
The overall construction, which is recursive, is summarized in \Cref{tab:semantics}.

\begin{restatable}{lemma}{endofunction}%
    \label{thm:endofunction}
    For every \redip program $\prog$, the automata transformer $\semantics{\prog}$ defined in \Cref{tab:semantics} is an endofunction on the set of normalized PGA over variables $\V$.
\end{restatable}
\begin{proof}[sketch]
    Follows by induction on the structure of $\prog$ since all constructions in \Cref{tab:semantics} preserve normalized PGA, which we detail in the following.
\end{proof}

\begin{table}[t]
    \centering
    \caption{
        Inductive definition of the translation from a prior distribution automaton $\Raut$ to an automaton $\semantics{\prog}(\Raut)$ for the (unnormalized) posterior distribution.
    }
    \label{tab:semantics}
    \setlength{\tabcolsep}{4pt}
    \rowcolors{2}{gray!7}{white} 
    \resizebox{\textwidth}{!}{
    \begin{tabular}{l l l}
        \toprule
        $\prog$& $\trans{\prog}{\Raut}$  & \textit{Construction} \\ \midrule
        $\assign{X}{0}$ & $\labelSubs{\Raut}{X}{1}$ & Label substitution \\
        $\incr{X}{n}$ & $\Raut \concat \Raut_{\dirac{n}{X}}$ & Concatenation \\
        $\incr{X}{D}$ & $\Raut \concat \Raut_{D_X}$ & Concatenation \\
        $\incr{X}{Y}$ & $\labelSubs{\Raut}{Y}{\inlineTransTwo{Y}{X}}$ & Transition substitution \\
        $\incr{X}{\iid{D}{Y}}$ & $\labelSubs{\Raut}{Y}{\inlineTrans{Y} \concat \Raut_{D_X}}$ & Transition substitution \\
        $\coinflip{\prog_1}{p}{\prog_2}$ & $\probchoice{\trans{\prog_1}{\Raut}}{\trans{\prog_2}{\Raut}}{p}{1{-}p}$ & Disjoint union \\ 
        $\ite{\guard}{\prog_1}{\prog_2}$ & $\trans{\prog_1}{\Raut \times\Baut_\guard} \oplus \trans{\prog_2}{\Raut \times\Baut_{\neg\guard}}$ & Products \& disjoint union \\
        $\decr{X}$ & $\Raut^{\decr{X}}$ & Special construction (Def.~\ref{def:decr})\\
        $\observe{\guard}$ & $\Raut \times \Baut_{\guard}$ & Product \\
        $\prog_1\fatsemi\prog_2$ & $\trans{\prog_2}{\trans{\prog_1}{\Raut}}$ & (None required) \\
        \bottomrule
    \end{tabular}
    }
\end{table}

\subsection{Setting a Variable to Zero: Transition Label Substitution}
\label{sec:setToZero}
The idea for applying $\assign{X}{0}$ to the distribution $\semanticsaut{\Raut}$ is as follows.
Each time $\Raut$ takes an $X$-transition, i.e.\ ``increases the $X$-counter'', the updated automaton $\semantics{\assign{X}{0}}(\Raut)$ must not increase that counter.
More formally, we define the automaton $\labelSubs{\Raut}{X}{1}$ which is almost identical to $\Raut$, except that all transitions of the form $\inlineTrans{rX}$ in $\Raut$ (recall $r \in \Rgez$) become ''$\varepsilon$-transitions'' $\inlineTrans{r}$ in $\labelSubs{\Raut}{X}{1}$.
For a normalized $\Raut$, the effect of this construction on the PGF $\semanticsaut{\Raut}$ is a substitution of $X$ by $1$; that is, $X$ is marginalized out.
See \Cref{fig:transLabSubs} for an example.

\begin{figure}[t]
    \centering
    \begin{minipage}{0.39\textwidth}
        \begin{tikzpicture}[node distance=7mm]
            \node at (-0.75,.5) {\footnotesize$\Raut\colon$};
            \node[state,initial] (s) {};
            \node[state,final,right=of s] (t) {};
            \draw[->] (s) -- node[auto] {$\tfrac 1 2 Y$} (t);
            \draw[->] (s) edge[loop above,looseness=5] node[auto,\myred] {$\tfrac 1 2 X$} (s);
        \end{tikzpicture}
        \caption*{$\semanticsaut{\Raut} = \tfrac 1 2 Y + \tfrac 1 4 X Y + \tfrac 1 8 X^2Y + \ldots$}
    \end{minipage}
    \begin{minipage}{0.6\textwidth}
        \begin{tikzpicture}[node distance=7mm]
            \node[align=right] at (-2.25,.5) {\footnotesize$\semantics{\assign{X}{0}}(\Raut) = \labelSubs{\Raut}{X}{1}\colon$};
            \node[state,initial] (s) {};
            \node[state,final,right=of s] (t) {};
            \draw[->] (s) -- node[auto] {$\tfrac 1 2 Y$} (t);
            \draw[->] (s) edge[loop above,looseness=5] node[auto,\myred] {$\tfrac 1 2$} (s);
        \end{tikzpicture}
        \caption*{$\semanticsaut{\labelSubs{\Raut}{X}{1}} = Y$}
    \end{minipage}
    \caption{Transition label substitution $\Raut[X/1]$ for implementing the instruction $\assign{X}{0}$.}
    \label{fig:transLabSubs}
\end{figure}

\subsection{Incrementing a Variable I: Automata Concatenation}
We now consider the statement $\incr{X}{n}$ and its generalization $\incr{X}{D}$ for incrementing $X$ by a constant and a random sample, 
respectively.

First, the automaton $\semantics{\incr{X}{n}}(\Raut)$ must increment the $X$-counter an additional $n$ times compared to $\Raut$.
To achieve this, we append a gadget---in fact, the PGA $\Raut_{\dirac{n}{X}}$ from \Cref{fig:basicDistrAutomata}---to the final states of $\Raut$.
Formally, this construction is the standard \emph{concatenation} $\Raut \concat \Raut_{\dirac{X}{n}}$, defined visually in \Cref{fig:concatUnion}.
For general $\semiring$-automata $\Raut_1,\Raut_2$ ($\semiring$ an $\omega$-continuous semiring), their concatenation satisfies $\semanticsaut{\Raut_1 \concat \Raut_2} = \semanticsaut{\Raut_1} \cdot \semanticsaut{\Raut_2}$; see, e.g.,~\cite[Thm.~4.6]{DBLP:reference/hfl/Kuich97}.

Second, and more generally, to construct $\semantics{\incr{X}{D}}(\Raut)$ we take the concatenation of $\Raut$ and $\Raut_{D_X}$, a PGA for distribution $D$ w.r.t.\ $X$. Thus, after reaching a final state of $\Raut$, the automaton $\Raut \concat \Raut_{D_X}$ additionally traverses $\Raut_{D_X}$, which has the desired effect of incrementing the sample from $\semanticsaut{\Raut}$ by a sample from $D$.  
We require the PGA for distribution $D$ to also be normalized and have a probability mass of 1.

\begin{figure}[t]
\centering
    \begin{minipage}{0.49\textwidth}
        \centering
        \resizebox{\textwidth}{!}{
			\begin{tikzpicture}

\node[draw, minimum height=4cm, minimum width=3cm] at (3,0) (A1) {};
\node at (A1.center) {\huge$\mathcal{A}_1$};

\coordinate (lu) at (0.5, 1.5) ;
\coordinate (ll) at (0.5,-1.5) ;

\coordinate (liu) at (2,1.5) ;
\coordinate (lil) at (2,-1.5) ;

\draw[edge] (lu) -- (liu) node[pos=0.2, above]{\Large$I_1$};
\draw[edge] (ll) -- (lil) node[pos=0.2, below]{\Large$I_n$};

\coordinate (ru) at (5.5, 1.5) ;
\coordinate (rl) at (5.5,-1.5) ;

\coordinate (lur) at (8, 1.5) ;
\coordinate (llr) at (8,-1.5) ;
\coordinate (riu) at (4,1.5) ;
\coordinate (ril) at (4,-1.5) ;

\draw[edge] (riu) -- (lur) node[midway, above]{\Large $F_1I_{1'}$};
\draw[edge] (ril) -- (llr) node[midway,below]{\Large$F_nI_{n'}$};

\draw[edge] (riu) -- (llr);
\draw[edge] (ril) -- (lur);


\node[draw, minimum height=4cm, minimum width=3cm] at (9,0) (A1) {};
\node at (A1.center) {\huge$\mathcal{A}_2$};

\coordinate (ru) at (11.5, 1.5) ;
\coordinate (rl) at (11.5,-1.5) ;

\coordinate (riu) at (10,1.5) ;
\coordinate (ril) at (10,-1.5) ;

\draw[edge] (riu) -- (ru) node[pos=0.8, above]{\Large$F_{1'}$};
\draw[edge] (ril) -- (rl) node[pos=0.8, below]{\Large$F_{n'}$};

\end{tikzpicture}
			}
    \end{minipage}
    \begin{minipage}{0.49\textwidth}
        \centering
        \resizebox{\textwidth}{!}{
			\begin{tikzpicture}
\node[draw, minimum height=4cm, minimum width=3cm] at (3,0) (A1) {};
\node at (A1.center) {\huge$\mathcal{A}_1$};

\coordinate (lu) at (0.5, 1.5) ;
\coordinate (ll) at (0.5,-1.5) ;

\coordinate (liu) at (2,1.5) ;
\coordinate (lil) at (2,-1.5) ;

\draw[edge] (lu) -- (liu) node[pos=0.2, above]{\Large$pI_1$};
\draw[edge] (ll) -- (lil) node[pos=0.2, below]{\Large$pI_n$};

\coordinate (ru) at (5.5, 1.5) ;
\coordinate (rl) at (5.5,-1.5) ;

\coordinate (riu) at (4,1.5) ;
\coordinate (ril) at (4,-1.5) ;

\draw[edge] (riu) -- (ru) node[pos=0.8, above]{\Large$F_1$};
\draw[edge] (ril) -- (rl) node[pos=0.8, below]{\Large$F_n$};


\node[draw, minimum height=4cm, minimum width=3cm] at (9,0) (A1) {};
\node at (A1.center) {\huge$\mathcal{A}_2$};

\coordinate (lu) at (6.5, 1.5) ;
\coordinate (ll) at (6.5,-1.5) ;

\coordinate (liu) at (8,1.5) ;
\coordinate (lil) at (8,-1.5) ;

\draw[edge] (lu) -- (liu) node[pos=0.2, above]{\Large$qI_{1'}$};
\draw[edge] (ll) -- (lil) node[pos=0.2, below]{\Large$qI_{n'}$};

\coordinate (ru) at (11.5, 1.5) ;
\coordinate (rl) at (11.5,-1.5) ;

\coordinate (riu) at (10,1.5) ;
\coordinate (ril) at (10,-1.5) ;

\draw[edge] (riu) -- (ru) node[pos=0.8, above]{\Large$F_{1'}$};
\draw[edge] (ril) -- (rl) node[pos=0.8, below]{\Large$F_{n'}$};

\end{tikzpicture}
	}
    \end{minipage}
    \caption{
        Concatenation $\Aut_1 \concat \Aut_2$ (left) and weighted disjoint union $\probchoice{\Raut_1}{\Raut_2}{p}{q}$ (right) of weighted automata $\Aut_1$ and $\Aut_2$ with disjoint sets of states $\{1,\ldots,n\}$ and $\{1',\ldots,n'\}$, respectively.
        We write $\probchoice{\Raut_1}{\Raut_2}{}{}$ instead of $\probchoice{\Raut_1}{\Raut_2}{1}{1}$.
    }
    \label{fig:concatUnion}
\end{figure}

\subsection{Incrementing a Variable II: Transition Substitution}
\label{sec:iid}%
Next, we consider $\incr{X}{Y}$.
The idea for implementing this statement is a refined variant of transition label substitution.
Intuitively, the automaton $\semantics{\incr{X}{Y}}(\Raut)$ should be similar to $\Raut$, but every time $\Raut$ increments $Y$, the new automaton increments both $Y$ \emph{and} $X$.
We thus replace every transition of the form $\inlineTrans{rY}$ in $\Raut$ by the sub-automaton\footnote{We omit initial and final weights here and assume implicitly, that the left-most state is initial and the right-most state is final (with a weight of 1).} $\inlineTransTwo{rY}{X}$.
Notice that this is also meaningful if $X=Y$.
We call this operation \emph{transition substitution}:

\begin{definition}[Transition Substitution]%
    For $\lin$-automata $\Raut,\Raut'$, we define the $\lin$-automaton $\labelSubs{\Raut}{X}{\Raut'}$ as in \Cref{fig:transSubs}.
\end{definition}
Observe that label substitution $\labelSubs{\Raut}{X}{1}$ from \Cref{sec:setToZero} is just a shorthand for the transition substitution $\labelSubs{\Raut}{X}{{\bigcirc}}$.

We now consider the special statement $\incr{X}{\iid{D}{Y}}$.
Its intended semantics is that of the program ``$\texttt{for}~i=1..Y~\texttt\{~\incr{X}{D}~\texttt\}$'', i.e.\ $X$ is incremented by a sum of $Y$ many i.i.d.\ samples from the distribution $D$.
Notice that $\incr{X}{Y}$ is equivalent to $\incr{X}{\iid{\dirac{1}{}}{Y}}$.
The motivation for including the rather specialized $\texttt{iid}$ statement as a primitive in \redip is that it is naturally supported on PGFs in closed form~\cite{redip} and, as we show here, on PGA.
Indeed, we can construct an automaton $\semantics{\incr{X}{\iid{D}{Y}}}(\Raut)$ which is again similar to $\Raut$, except that every time $\Raut$ traverses a $Y$-transition, the new automaton additionally traverses $\Raut_{D_X}$ as a sub-automaton.
This amounts to the following transition substitution: $\semantics{\incr{X}{\iid{D}{Y}}}(\Raut) = \labelSubs{\Raut}{Y}{\inlineTrans{Y} \concat \Raut_{D_X}}$.

The \texttt{iid} statement is also of practical interest:
For example, it holds that $\binomial{N}{p} = \texttt{iid}(\bern{p}{}, N)$, $\negbinom{N}{p} = \texttt{iid}(\geom{p}{}, N)$, etc.
Therefore, \redip allows sampling from these distributions even if their parameter $N$ is a program variable, which may itself be non-trivially distributed.

\begin{figure}[t]
    \centering
    \begin{tikzpicture}[node distance=1.5cm]
    \node[state] (s) {$s$};
    \node[phantom, left=0.25cm of s,label={180:$\Raut$:}] (ls) {};
    \node[state, right of=s] (t) {$t$};

    \draw[edge] (s) --  node[auto]{$rX$} (t);

    \node[state] at (5,0) (ns) {$s$};
    \node[phantom, left=0.25cm of ns,label={180:$\Raut[X/\Raut']$:}] (lns) {};
     \node[draw, minimum height=2cm, minimum width=1.5cm] at (7.5,0) (A') {};
  \node at (A'.center) {$\mathcal{A}'$};
  \node[state] at (10,0) (nt) {$t$};
  
\node[phantom] at(6.5,0.1) {$\vdots$};
\node[phantom] at(8.5,0.1) {$\vdots$};
  
  \coordinate (ulA) at (7,0.5);
	\coordinate (blA) at (7,-0.5);
  \coordinate (urA) at (8,0.5);
  \coordinate (brA) at (8,-0.5);
  
  \draw[edge] (ns) -- node[above]{$rI_1$} (ulA);
  \draw[edge] (ns) -- node[below]{$rI_n$}(blA);
  \draw[edge] (urA) -- node[above]{$F_1$}(nt);
\draw[edge] (brA) -- node[below]{$F_m$}(nt);
\end{tikzpicture}\vspace{-2em}
    \caption{Transition substitution. The automaton $\labelSubs{\Raut}{X}{\Raut'}$ arises from $\Raut$ by replacing \emph{all} transitions $\inlineTrans{rX}$ in $\Raut$ (left) by the gadget depicted on the right.}
    \label{fig:transSubs}
\end{figure}

\subsection{Random Branching: Disjoint Union}
To implement the random branching instruction $\coinflip{\prog_1}{p}{\prog_2}$, we construct an automaton that, intuitively speaking, behaves like $\semantics{\prog_1}(\Raut)$ with probability $p$ and like $\semantics{\prog_2}(\Raut)$ with probability $1-p$.
This is achieved by multiplying each initial weight of $\semantics{\prog_1}(\Raut)$ and $\semantics{\prog_2}(\Raut)$ by $p$ and $1-p$, respectively, and taking the \emph{disjoint union} of the resulting two automata (in symbols: $\probchoice{\semantics{\prog_1}(\Raut)}{\semantics{\prog_2}(\Raut)}{p}{1-p}$; see \Cref{fig:concatUnion}).
The automata $\semantics{\prog_1}(\Raut)$ and $\semantics{\prog_2}(\Raut)$ themselves are obtained recursively.

\subsection{Conditional Branching: Product with Guard Automata}
\label{sec:ite}
We now consider the statement $\ite{\guard}{\prog_1}{\prog_2}$.
The basic idea is to ``partition'' the distribution $\semanticsaut{\Raut}$ relative to $\guard$ and $\neg\guard$, recursively transform the resulting (sub-)distributions---the first by $\semantics{\prog_1}$ and the second by $\semantics{\prog_2}$---and finally sum up the results~\cite{DBLP:journals/jcss/Kozen81}. To realize the partitioning on automaton-level, we employ the DFA $\Baut_\guard$ and $\Baut_{\neg\guard}$ from \Cref{sec:programs}, and take the products $\Raut \times \Baut_\guard$ and $\Raut \times \Baut_{\neg\guard}$.
To this end, we rely on the following product construction, a weighted variant of the standard product of an $\varepsilon$-NFA with a (non-$\varepsilon$) NFA:

\begin{definition}[Product]%
    Let $\Raut = (\states, M, I, F)$ be a $\lin$-automaton and let $\Baut = (\states', M', I', F')$ be an NFA with alphabet $\V$.
    We define the \emph{product} $\Raut \times \Baut$ as the $\lin$-automaton with states $\states \times \states'$ and transitions according to the following rules:
    \begin{align*}
        \sosrule{$X$-trans.}{\Raut \colon q \xrightarrow{a X} r \qquad \Baut \colon s \xrightarrow[]{X} t}{\Raut \times \Baut \colon (q,s) \xrightarrow[]{aX} (r,t)}
        \qquad
        \sosrule{$\varepsilon$-trans.}{\Raut \colon q \xrightarrow[]{a} r \qquad s \in \states'}{\Raut \times \Baut \colon (q,s) \xrightarrow[]{a} (r,s)}
    \end{align*}
    The initial and final weight of $(q,s) \in \states \times \states'$ is $I_q \cdot I_s$ and $F_q \cdot F_s$, respectively.
\end{definition}
If $\Raut$ is normalized, then so is $\Raut \times \Baut$.
The next lemma asserts that the products $\Raut \times \Baut_\guard$ and $\Raut \times \Baut_{\neg\guard}$ have the desired effect (recall that $\Baut_\guard$ and $\Baut_{\neg\guard}$ are closed under permutations by definition):

\begin{restatable}{lemma}{filter}%
    \label{thm:filter}%
    Let $\Raut$ be a normalized $\lin$-automaton over $\fps{\Rgezinf}{\N^\V}$ and let $\Baut$ be a DFA\footnote{\Cref{thm:filter} holds more generally if $\Baut$ is only unambiguous instead of deterministic.} over alphabet $\V$.
    Assume that the language accepted by $\Baut$ is closed under permutations.
    Then, for all $\val \in \N^\V$, it holds that
    \[
        \semanticsaut{\Raut \times \Baut}(\sigma)
        ~=~
        \begin{cases}
            \semanticsaut{\Raut}(\val) & \text{ if $\Baut$ accepts some $w \in V^*$ with $\Psi(w) = \val$,} \\
            0 & \text{ else.}
        \end{cases}
    \]
\end{restatable}
\noindent To complete the construction for $\ite{\guard}{\prog_1}{\prog_2}$, we obtain $\semantics{\prog_1}(\Raut \times \Baut_\guard)$ and $\semantics{\prog_2}(\Raut \times \Baut_{\neg\guard})$ recursively and recombine them via disjoint union (see \Cref{tab:semantics}).

\subsection{Decrementing a Variable: Specialized Construction}
The decrement statement $\decr{X}$ is somewhat involved since our program variables are non-negative.
Following~\cite{redip}, we adopt the convention that decrementing $0$ has no effect, i.e.\ $\decr{X}$ is equivalent to $\ite{X>0}{\decr{X}}{\skp}$.
Based on this, we propose to implement $\decr{X}$ similar to \texttt{if-else} from \Cref{sec:ite} as follows:
We first filter $\Raut$ w.r.t.\ $X > 0$; that is, we take the product $\Raut \times \Baut_{X>0}$.
Afterwards, we convert \emph{the first} $X$-transition on any path of this product to an $\varepsilon$-transition.
This is achieved by a transition label substitution, which is, however, only applied to the transitions of the form $(\ldots,s) \xrightarrow{rX} (\ldots,t)$, where $s \neq t$ are the two states of $\Baut_{X>0}$.
Finally, we take the disjoint union with $\Raut \times \Baut_{X=0}$, which has the same behavior as $\labelSubs{\Raut}{X}{0}$. See \Cref{fig:decr} for an example.
Formally:

\begin{definition}[Decrement Automaton]%
    \label{def:decr}%
    Let $\Raut$ be a normalized PGA and let $s,t$ be the two states of $\Baut_{X>0}$.
    Consider the PGA $\Raut' = (\Raut \times \Baut_{X>0}) \oplus \labelSubs{\Raut}{X}{0}$.
    The \emph{decrement PGA} $\Raut^{\decr{X}}$ is obtained from $\Raut'$ via replacing every transition $(q,s) \xrightarrow{rX} (u,t)$ by $(q,s) \xrightarrow{r} (u,t)$, where $q,u$ are arbitrary states of $\Raut$, and $r \in \Rgez$.
\end{definition}

\begin{figure}[t]
    \centering
    \begin{tikzpicture}[node distance=2mm and 7mm]
        \node at (-1.25,0.5) {$\Raut_{\geom{\tfrac 1 2}{X}}\colon$};
        \node[state,labeledfinal=$\tfrac 1 2$,initial] (q0) {$q$};
        \draw[loop above,looseness=5] (q0) to node[above] {$\tfrac 1 2 X$} (q0);
    \end{tikzpicture}
    \qquad
    \begin{tikzpicture}[node distance=2mm and 12mm]
        \node at (-1.25,0.5) {$\Baut_{X>0}\colon$};
        \node[state,initial] (s) {\red{$s$}};
        \node[state,final,right=of s] (t) {\red{$t$}};
        \draw[loop above,looseness=5] (s) to node[above] {$\Y$} (s);
        \draw[loop above,looseness=5] (t) to node[above] {$X,\Y$} (t);
        \draw[->] (s) to node[above] {$X$} (t);
    \end{tikzpicture}
    \\
    \begin{tikzpicture}[node distance=2mm and 12mm]
        \node at (-1.25,0.5) {$\Raut^{\decr{X}}_{\geom{\tfrac 1 2}{X}}\colon$};
        \node[state,initial] (s) {$q,\red{s}$};
        \node[state,labeledfinal=$\tfrac 1 2$,initial,right=of t] (other) {};
        \node[state,labeledfinal=$\tfrac 1 2$,right=of s] (t) {$q,\red{t}$};
        \draw[loop above,looseness=5] (t) to node[above] {$\tfrac 1 2 X$} (t);
        \draw[->] (s) to node[above] {$\tfrac 1 2$} (t);
    \end{tikzpicture}
    \caption{Example for the decrement construction from  \Cref{def:decr} (with $\Y = \V\setminus\set{X}$).}
    \label{fig:decr}
\end{figure}
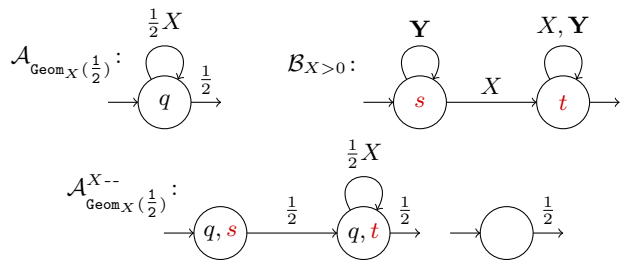

\subsection{Conditioning (\texttt{observe}): Product}
\label{sec:observe}
The purpose of the $\observe{\guard}$ statement is to discard the portion of the input distribution $\semanticsaut{\Raut}$ violating $\guard$.
This may yield a proper sub-distribution, the \emph{unnormalized posterior}.
Implementing $\observe{\guard}$ on PGA-level works similarly to conditional branching (\Cref{sec:ite}):
We define $\semantics{\observe{\guard}}(\Raut) = \Raut \times \Baut_\guard$, i.e.\ we filter the distribution $\semanticsaut{\Raut}$ according to $\guard$.

A PGA $\normalize{\Raut}$ for the \emph{normalized} distribution of any arbitrary PGA $\Raut$---this applies in particular to the unnormalized posterior PGA $\semantics{\prog}{(\Raut)}$ resulting from a program $\prog$---can be effectively constructed as follows:
Let $M_{\Raut}$ be the probability mass of the (sub-)distribution $\semanticsaut\Raut$, which can be determined using \Cref{thm:computeMass}, and let
\[
    \normalize{\Raut}
    ~=~
    \begin{cases}
        \text{$\Raut$ with initial weights multiplied by $\tfrac{1}{M_\Raut}$} & \text{if $M_\Raut > 0$} ~, \\
        \text{undefined} & \text{if $M_\Raut = 0$} ~.
    \end{cases}
\]
If defined, it follows that $\semanticsaut{\normalize{\Raut}} = \tfrac{1}{M_\Raut} \semanticsaut{\Raut}$, as desired.
Notice that we could have defined $\normalize{\cdot}$ equivalently by scaling up the final instead of the initial weights.

\begin{example}%
    \label{ex:semanticsMotivating}%
    \Cref{fig:motivatingExample} shows a program and the resulting normalized posterior PGA\footnote{Unreachable states have been removed and some transitions are slightly simplified.} with respect to a prior Dirac distribution where all variables are $0$ with probability $1$.
    The normalizing constant $\tfrac{11}{40}$ can be computed ``by hand'' (i.e.\ without solving an explicit LP) as follows:
    We consider the upper and the lower branch of the automaton in \Cref{fig:motivatingExample} separately.
    The upper branch has a mass of
    \[
        \tfrac{9}{10} \cdot \tfrac{1}{2} \cdot \tfrac{1}{2} \cdot (\tfrac{1}{2})^* \cdot \tfrac{1}{2}
        ~=~
        \tfrac{9}{40}
    \]
    whereas the lower branch has mass
    \[
        \tfrac{1}{10} \cdot \big(\tfrac{1}{2}\cdot\tfrac{1}{2}\cdot(\tfrac{1}{2})^*\cdot\tfrac{1}{2}\cdot(\tfrac{1}{2})^* \cdot\tfrac{1}{2} + 2 \cdot(\tfrac{1}{2}\cdot\tfrac{1}{2} \cdot\tfrac{1}{2}\cdot(\tfrac{1}{2})^*\cdot\tfrac{1}{2}) \big)
        ~=~
        \tfrac{1}{20}
        ~.
    \]
    Here, we have used that $(\tfrac 1 2)^* = \sum_{n = 0}^{\infty} (\tfrac 1 2)^n = 2$. 
    The mass of the unnormalized posterior PGA, i.e.\ the normalizing constant, is therefore $\tfrac{9}{40} + \tfrac{1}{20} = \tfrac{11}{40}$.
\end{example}

\subsection{Complexity}

We conclude this section by examining the overall complexity of the construction resulting from \Cref{tab:semantics}.
We define the \emph{size} $\sizeaut{\Raut}$ of a PGA $\Raut$ as the number of transitions with non-zero weight (or $\sizeaut{\Raut} = 1$, if $\Raut$ has no transitions),
and the size of a guard $\guard$ (see \Cref{tab:guards}) recursively as $\sizebool{\guard_1 \land \guard_2} = \sizebool{\guard_1} + \sizebool{\guard_2}$, $\sizebool{\neg\guard} = \sizebool{\guard}$, and $\sizebool{\guard} = 1$ for the other two atomic cases.
Similarly, the size $\sizeprog{\prog}$ of a \redip program $\prog$ is defined recursively as $\sizeprog{\prog_1\fatsemi\prog_2} = \sizeprog{\prog_1} + \sizeprog{\prog_2}, \sizeprog{\ite{\guard}{\prog_1}{\prog_2}} = \sizeprog{\coinflip{\prog_1}{p}{\prog_2}} = 1 + \sizeprog{\prog_1} + \sizeprog{\prog_2}$, and $\sizeprog{\prog} = 1$ for all base cases. 
Intuitively, the size of a program is roughly proportional to the length of the program text.

\begin{restatable}[Complexity]{theorem}{complexitysemantics}%
    \label{thm:complexity-semantics}%
    For every \redip\ program $\prog$ and PGA $\Raut$ we have
    \[
    \sizeaut{\semantics{\prog}(\Raut)}
    ~\in~
    \bigO{\sizeaut{\Raut} \cdot \mathcal{D}^{\sizeprog{\prog}} \cdot (n+1)^{\sizeprog{\prog}\eta}}
    \]
    where $\mathcal{D} = \max\set{\sizeaut{\Raut_D} \colon \Raut_D \text{ distribution automaton in }\prog}$ is the size of the largest distribution automaton, $n  =\max\set{m \colon m \text{ constant\footnote{Recall from \Cref{def:syntax-redip} that integer constants can occur on the right hand side of increment instructions and in guards.} in }\prog}$ is the largest constant and $\eta = \max\set{\sizebool{\guard}\colon \guard \text{ guard in }\prog}$ is the largest guard in $\prog$ (with the convention that $\max\emptyset = 1$).
\end{restatable}
\noindent Mitigating the exponential worst-case complexity by minimizing the intermediate PGA is a promising direction for future work.
\section{Soundness Relative to Operational Semantics}
\label{sec:operational}
In this section, we prove that our PGA transformations from \Cref{tab:semantics} are sound with respect to an existing (small-step) operational semantics for imperative probabilistic programs as presented in~\cite{DBLP:journals/toplas/OlmedoGJKKM18}.
This semantics is defined in terms of a (discrete-time) \emph{Markov chain}, where the states are tuples $\opstate{\prog}{\val}$, containing the program statement to be executed next and the current variable valuation.
Additionally, there are some special states, which are detailed below.
The transitions of the Markov chain model individual execution steps.
An excerpt of the rules defining this Markov chain is given in \Cref{fig:opsem}.
We remark that the straightforward rule \ruletag{Sample} was not included in~\cite{DBLP:journals/toplas/OlmedoGJKKM18}; we have incorporated it to model the \redip statement $\incr{x}{D}$, which increments the variable $x$ by a random sample from a distribution $D$ over $\N$. In the subsequent discussion, we disregard programs containing the special \texttt{iid} instruction because its operational semantics is somewhat involved.
In fact, the \texttt{iid} statement has not been discussed in prior work on operational semantics~\cite{DBLP:journals/toplas/OlmedoGJKKM18}.
A complete construction is left for future work.

Formally, a \emph{Markov chain} can be defined as a $[0,1]$-weighted automaton $\mathcal{M} = (\markovstates,\markovtrans,\markovinit)$ without final weights where the initial weights $\markovinit$ sum to at most one, and where the matrix $\markovtrans$ is (row-)stochastic.

\begin{definition}[Markov Chain Semantics]\label{def:operationalMarkovChain}%
    Let $\Raut$ be a PGA (for the initial distribution) and let $\prog$ be a \redip program without $\mathtt{iid}$ statements.
    We define the \emph{operational Markov chain} $\markovchain{\Raut}{\prog} = (\markovstates,\markovtrans,\markovinit)$ as follows:
    \begin{itemize}
        \item $\markovstates ~=~ \big(\text{\textnormal\redip} \cup \{\downarrow\}\big) \times\N^V ~\cup~ \{\langle\lightning\rangle\}$
        \item $\markovinit(s) ~=~ \begin{cases}
            \semanticsaut{\Raut}(\val) & \text{ if }s = \opstate{\prog}{\val}~, \\
            0 & \text{ else}
        \end{cases}$
        \item $\markovtrans(s,t) ~=~ \begin{cases}
            p & \text{ if } s \xrightarrow{p} t \text{ is derivable from the rules in \Cref{fig:opsem}}, \\
            0 & \text{ else}
        \end{cases}$
    \end{itemize}
\end{definition}
States of the form $\opstate{\downarrow}{\val}$ indicate that the program has \emph{terminated} in state $\val$; the special absorbing state $\langle\lightning\rangle$ is entered once an observation violation occurs.
%
We can now state our soundness theorem formally.

\begin{restatable}[Soundness w.r.t.\ Operational Semantics]{theorem}{opequiv}\label{thm:opequiv}%
    Let $\Raut$ be a PGA and let $\prog$ be a \redip program (without $\mathtt{iid}$ statements) such that $\normalize{\semantics{\prog}(\Raut)}$ is defined.
    Then, for all $\val\in\N^V$, it holds that
    \[
        \semanticsaut{\normalize{\semantics{\prog}(\Raut)}}(\val)
        ~=~
        \markovprob{\prog}{\Diamond\opstate{\downarrow}{\val} \mid \neg \Diamond\langle\lightning\rangle}
        ~,
    \]
    where the right hand side is the \emph{conditional} probability of reaching state $\opstate{\downarrow}{\val}$ in $\markovchain{\Raut}{\prog}$ given that $\langle\lightning\rangle$ is not reached (see \cite{DBLP:journals/toplas/OlmedoGJKKM18} for the formal definition).
\end{restatable}

\begin{example}[Operational Markov Chain]%
    \label{ex:mc-semantics}%
    Consider the program
    \[
        \prog
        \quad=\quad
        \overbrace{\{\underbrace{\incr{X}{Y}}_{\prog_{11}}\} \,[\nicefrac{1}{2}]\, \{\underbrace{\skp}_{\prog_{12}}\}}^{\prog_1}
        ~\fatsemi~
        \underbrace{\observe{X=0}}_{\prog_2}
        ~.
    \]
    Assuming an initial PGA $\Raut$ with $\semanticsaut{\Raut} = \frac{1}{2} + \frac{1}{2}Y^2$, the reachable fragment of the resulting operational Markov chain $\markovchain{\Raut}{\prog}$ is depicted in \Cref{fig:opsemantics}.
    We have:
    \begin{itemize}
        \item $\markovprob{\prog}{\Diamond\opstate{\downarrow}{\val_{00}}} = \frac{1}{2} = \semanticsaut{\semantics{\prog}(\Raut)}(\val_{00})$
        \item $\markovprob{\prog}{\Diamond\opstate{\downarrow}{\val_{02}}} = \frac{1}{4} = \semanticsaut{\semantics{\prog}(\Raut)}(\val_{02})$
        \item $\markovprob{\prog}{\Diamond\errorstate} = \frac{1}{4} = 1 - \frac{3}{4} = \sum_{\sigma\in\N^V} \semanticsaut{\Raut}(\sigma) - \sum_{\sigma\in\N^V}\semanticsaut{\semantics{\prog}(\Raut)}(\sigma)$
        \item $\markovprob{\prog}{\Diamond\opstate{\done}{\val_{00}}~\vert~\neg\Diamond\errorstate} = \frac{\nicefrac{1}{2}}{\nicefrac{3}{4}} = \frac{2}{3} = \semanticsaut{\normalize{\semantics{\prog}(\Raut)}}(\val_{00})$
        \item $\markovprob{\prog}{\Diamond\opstate{\done}{\val_{02}}~\vert~\neg\Diamond\errorstate} = \frac{\nicefrac{1}{4}}{\nicefrac{3}{4}} = \frac{1}{3} = \semanticsaut{\normalize{\semantics{\prog}(\Raut)}}(\val_{02})$
    \end{itemize}
\end{example}

\begin{figure}[t]
    \centering
    \addtolength{\jot}{1em} 
    \begin{gather*}
        \sosrule{Sample}{D \in Distr(\N) \qquad D(n) = p}{\opstate{\incr{X}{D}}{\val} ~\xrightarrow{p}~ \opstate{\done}{\val[X \gets \val(X) + n]}}
        \qquad
        \ruletag{Obs-f} \frac{\sigma\not\models\guard}{\opstate{\observe{\guard}}{\val}~\rightarrow~\langle\stmaryrdLightning\rangle }\\
        \ruletag{Choice-l} \frac{}{\opstate{\coinflip{\prog_1}{p}{\prog_2}}{\val}~\xrightarrow{p}~\opstate{\prog_1}{\val}}
        \qquad
        ...
    \end{gather*}
    \caption{Excerpt of the operational Markov chain semantics (see \Cref{sec:op-equiv} for full details). $\textsc{Sample}$ is a new rule added for this paper.}
    \label{fig:opsem}
\end{figure}

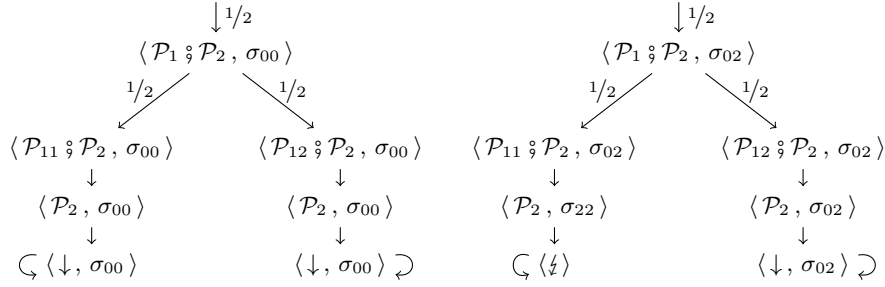
\begin{figure}[t]
    \centering
    \begin{tikzpicture}
\node[labeledinitial top={$\nicefrac{1}{2}$}] (init00) {$\opstate{\prog_1\fatsemi\prog_2}{\val_{00}}$};

\node[xshift=-0.5cm, yshift=-1cm] (p100) at (init00.south west) {$\opstate{\prog_{11}\fatsemi\prog_2}{\val_{00}}$};

\node[xshift=0.5cm, yshift=-1cm] (p200) at (init00.south east)  {$\opstate{\prog_{12}\fatsemi\prog_2}{\val_{00}}$};

\draw[edge] (init00) -- node[left, above, pos=0.7]{$\nicefrac{1}{2}$} (p100);
\draw[edge] (init00) -- node[right, above, pos=0.7]{$\nicefrac{1}{2}$}(p200);

\node[yshift=-0.5cm] (p2p100) at (p100.south) {$\opstate{\prog_2}{\val_{00}}$};

\node[yshift=-0.5cm] (p2p200) at (p200.south) {$\opstate{\prog_2}{\val_{00}}$};

\draw[edge] (p100) -- (p2p100);
\draw[edge] (p200) -- (p2p200);

\node[yshift=-0.5cm] (dp2p100) at (p2p100.south) {$\opstate{\done}{\val_{00}}$};
\node[yshift=-0.5cm] (dp2p200) at (p2p200.south) {$\opstate{\done}{\val_{00}}$};

\draw[edge] (p2p100) -- (dp2p100);
\draw[edge] (p2p200) -- (dp2p200);

\draw[->] (dp2p100) edge[loop left, out=-190, in=-170,looseness=3] (dp2p100);
\draw[->] (dp2p200) edge[loop left, out=10, in=-10,looseness=3] (dp2p200);


\node[xshift=5cm, labeledinitial top={$\nicefrac{1}{2}$}] (init02) at (init00.east) {$\opstate{\prog_1\fatsemi\prog_2}{\val_{02}}$};

\node[xshift=-0.5cm, yshift=-1cm] (rp100) at (init02.south west) {$\opstate{\prog_{11}\fatsemi\prog_2}{\val_{02}}$};

\node[xshift=0.5cm, yshift=-1cm] (rp200) at (init02.south east)  {$\opstate{\prog_{12}\fatsemi\prog_2}{\val_{02}}$};

\draw[edge] (init02) -- node[left, above, pos=0.7]{$\nicefrac{1}{2}$} (rp100);
\draw[edge] (init02) -- node[right, above, pos=0.7]{$\nicefrac{1}{2}$}(rp200);

\node[yshift=-0.5cm] (rp2p100) at (rp100.south) {$\opstate{\prog_2}{\val_{22}}$};

\node[yshift=-0.5cm] (rp2p200) at (rp200.south) {$\opstate{\prog_2}{\val_{02}}$};

\draw[edge] (rp100) -- (rp2p100);
\draw[edge] (rp200) -- (rp2p200);

\node[yshift=-0.5cm] (rdp2p100) at (rp2p100.south) {$\langle\stmaryrdLightning\rangle$};
\node[yshift=-0.5cm] (rdp2p200) at (rp2p200.south) {$\opstate{\done}{\val_{02}}$};

\draw[edge] (rp2p100) -- (rdp2p100);
\draw[edge] (rp2p200) -- (rdp2p200);

\draw[->] (rdp2p100) edge[loop left, out=-200, in=-160,looseness=3] (rdp2p100);
\draw[->] (rdp2p200) edge[loop left, out=10, in=-10,looseness=3] (rdp2p200);

\end{tikzpicture}
    \caption{Operational Markov chain of the program $\prog$ from \Cref{ex:mc-semantics}.
        $\val_{ij}$ refers to a variable valuation $\val$ satisfying $\val(X) = i$ and $\val(Y) = j$.}
    \label{fig:opsemantics}
\end{figure}

%

\section{Related Work} 
\label{sec:relwork}%
The general idea of probabilistic programming as an engineer-friendly paradigm for encoding statistical inference problems goes back to at least~\cite{DBLP:conf/aaai/KollerMP97}.
As mentioned in \Cref{sec:intro}, our work is primarily inspired by more recent research focused on the analysis of probabilistic programs and exact inference using probability generating functions (PGFs)~\cite{generatingfunctionsfor,redip,zaiser,credip}.
Our approach can be seen as an automata-theoretic counterpart to these methods.

Several alternative techniques for exact inference have been proposed. The tool \texttt{Dice}~\cite{scaling-exact-inference} leverages weighted model counting, but it is not directly applicable to programs that sample from distributions with infinite support.
The \texttt{PSI} tool~\cite{psi} employs symbolic representations of probability density functions to perform inference.
In~\cite{SPPL}, sum-product expressions are used to compactly represent complex distributions.
\emph{Weakest pre-expectations}~\cite{DBLP:conf/stoc/Kozen83,DBLP:series/mcs/McIverM05}, a probabilistic generalization of classic weakest preconditions, were amended to conditioning and inference in~\cite{DBLP:journals/toplas/OlmedoGJKKM18}.
The more recent paper~\cite{recursive} focuses on exact inference in programs with recursion by compiling them to systems of polynomial equations.
Probabilistic model checking techniques for computing exact conditional probabilities exist as well~\cite{DBLP:conf/tacas/BaierKKM14}, but these methods are generally restricted to programs unwinding to \emph{finite} Markov chains.

Finally, we mention that there exist mature probabilistic programming systems such as \texttt{Stan}~\cite{JSSv076i01}, \texttt{WebPPL}~\cite{dippl}, and others.
These tools typically focus on \emph{approximate} inference using Monte-Carlo sampling methods.

\section{Conclusion and Future Work}
\label{sec:conc}
We have introduced a method for performing exact inference within a class of syntactically restricted discrete probabilistic programs with infinite support distributions.
The core idea is to represent a program's posterior distribution as a weighted automaton derivable from the program text via standard automata-theoretic constructions.

For future work, we aim to explore the possibility of relaxing some of the syntactic restrictions imposed by the \redip language.
Notably, \cite{zaiser} has recently extended the PGF-based approach from~\cite{redip,credip} to support certain \emph{continuous} distributions.
Investigating how such extensions could be integrated into our PGA encoding presents an intriguing direction.
Another promising avenue is to enrich the PGA model with more expressive automata-theoretic features, such as stacks or registers.
Given the extensive body of work in automata theory, we believe this framework is well-suited for characterizing broader classes of probabilistic programs where inference---and possibly other analysis problems---remain decidable.

We also plan to implement our approach, potentially building on existing probabilistic model checkers such as \texttt{Storm}~\cite{DBLP:journals/sttt/HenselJKQV22} or \texttt{PRISM}~\cite{DBLP:conf/cav/KwiatkowskaNP11}.
After all, the PGA resulting from our construction often resembles finite-state Markov chains, a class of models on which these tools perform very well~\cite{DBLP:conf/isola/BuddeHKKPQTZ20}.


%
%
\bibliographystyle{splncs04}
\bibliography{literature}

\newpage
\appendix
\allowdisplaybreaks 
\section{Proofs}
\subsection{Proof of \Cref{thm:filter}}
\filter*
\begin{proof}
    Let $\val \in \N^\V$ be arbitrary.
    There are two cases:
    \begin{enumerate}
        \item $\Baut$ accepts some $w \in \V^*$ with $\Psi(w) = \val$.
        This means that $\Baut$ in fact accepts \emph{all} such $w$, since $\Baut$ is closed under permutations by assumption.
        Now suppose that $(q_0,\ldots,q_n) \in Q^+$ is a finite path in $\Raut$ with weight
        \[
        I_{q_0} \cdot M_{q_0,q_1} \cdot \ldots \cdot  M_{q_{n-1}, q_n} \cdot F_{q_n} ~=~ a \vec{X}^\sigma
        ~,
        \]
        where $a \in \Rgez$.
        We have to show that we have \emph{exactly one path} of the form $((q_0,..), \ldots, (q_n,..))$ in $\Raut \times \Baut$ with the same weight $a \vec{X}^\sigma$.
        To determine this path, we choose $w$ as the sequence of symbols along $(q_0,\ldots,q_n)$, this time forgetting about the real weights but not about the order.
        Notice that $|w| = m \leq n$ and $m = 0$ is possible.
        Let $(s_0,\ldots,s_m) \in Q'^+$ be the unique accepting path of $\Baut$ on $w$.
        We can then construct the desired path as $((q_0,s_0),\ldots,(q_n,s_m))$ where, for all $i \in {1,\ldots,n}$, we set $s_i = s_{i-1}$ if $\Raut$ makes an $\varepsilon$-transition from $q_{i-1}$ to $q_i$;
        it can be checked that the weight of this path is $a \vec{X}^\val$ by construction of the product.
        \item $\Baut$ accepts no $w \in \V^*$ with $\Psi(w) = \val$.
        Then it follows by construction of the product that all paths on $\Raut \times \Baut$ with weight $a \vec{X}^\sigma$ (such paths are inscribed with some $w$ satisfying $\Psi(w) = \val$) must end in a state $(q,s)$ where $s$ is \emph{not} accepting, i.e.\ $F_s = 0$, and therefore $a = 0$.
    \end{enumerate}
\end{proof}

\subsection{Proof of \Cref{thm:endofunction}}
The proof of \Cref{thm:endofunction} relies on the following auxiliary lemma:

\begin{lemma}\label{lemma:properties-constructions}
Let $X,Y\in V, n \in\N, i\in\set{0,1}, p,q \in [0,1]$ and $\guard$ a Boolean guard. Then for normalized PGAs $\Raut,\Raut_1, \Raut_2$ the following holds:
\begin{enumerate}
\item $\Raut[X/i]$ is normalized PGA
\item $\Raut[Y/\Raut_2]$ is a normalized PGA
\item $\Raut_1 \concat \Raut_2$ is a normalized PGA
\item $\probchoice{\Raut_1}{\Raut_2}{p}{q}$ is normalized (and PGA, if $\probmass{\Raut_1} + \probmass{\Raut_2}\leq 1$ or if $p + q \leq 1$)
\item $\Raut \times \Baut_\guard$ is a normalized PGA
\item $\Raut^{\decr{X}}$ is a normalized PGA
\end{enumerate}
\end{lemma}
\begin{proof}
Let $\Raut,\Raut_1,\Raut_2$ normalized PGA, $X,Y \in V, n \in\N, i \in\set{0,1},p,q\in [0,1]$ and $\guard$ a Boolean guard. Recall that a $\lin$-automaton $\Raut = (Q,M,I,F)$ is called normalized if $\bigcup_{s\in Q} \supp{I_s}\cup \supp{F_s} \subseteq\set{1}$, i.e.\ $I$ and $F$ are $\Rgez$-vectors, and PGA if $\sum_{\sigma\in\N^V} \semanticsaut{\Raut}(\sigma)\leq 1$.

\noindent\textit{Case 1}: Since $\Raut$ is normalized, the substitution does not change the initial or final weights. Thus, $\Raut[X/i]$ is also normalized.

With $i = 1$ we do not change the probability mass of $\Raut$. With $i = 0$, we potentially remove probability mass from $\Raut$. Since our PGA-definition also includes sub-distributions, $\Raut[X/i]$ is also a PGA for both values of $i$.

\noindent\textit{Case 2}: For normalization, see Case 1. As for the PGA-property we have
\begin{align*}
\probmass{\Raut_1[Y/\Raut_2]} &= \sum_{\sigma\in\N^V,\sigma(Y) = 0} \semanticsaut{\Raut_1}(\sigma) + \\
&\quad \sum_{\sigma\in\N^V, \sigma(Y) \neq 0} \semanticsaut{\Raut_1}(\sigma)\prod_{\sigma(Y)} \left(\sum_{\sigma'\in\N^V}\semanticsaut{\Raut_2}(\sigma')\right) \\
&\leq \sum_{\sigma\in\N^V,\sigma(Y) = 0} \semanticsaut{\Raut_1}(\sigma) + \sum_{\sigma\in\N^V, \sigma(Y) \neq 0} \semanticsaut{\Raut_1}(\sigma)\prod_{\sigma(Y)} 1\\
&= \sum_{\sigma\in\N^V,\sigma(Y) = 0} \semanticsaut{\Raut_1}(\sigma) + \sum_{\sigma\in\N^V, \sigma(Y) \neq 0} \semanticsaut{\Raut_1}(\sigma) \\
&= \sum_{\sigma\in\N^V} \semanticsaut{\Raut}(\sigma) \\
&\leq 1
\end{align*}
\noindent\textit{Case 3}: The new initial (final) weights of $\Raut_1\concat\Raut_2$ are those of $\Raut_1$ ($\Raut_2$). Since both are normalized, so is $\Raut_1\concat\Raut_2$. 

For the PGA-property we have
\begin{align*}
\probmass{\Raut_1\concat\Raut_2} &= \sum_{\sigma\in\N^V}\semanticsaut{\Raut_1}(\sigma)\left(\sum_{\sigma'\in\N^V}\semanticsaut{\Raut_2}(\sigma')\right) \\
&\leq \sum_{\sigma\in\N^V}\semanticsaut{\Raut_1}(\sigma) \\
&\leq 1
\end{align*}

\noindent\textit{Case 4}: Since $p,q\in [0,1]$ and $\Raut_1,\Raut_2$ are normalized, the normalization of $\probchoice{\Raut_1}{\Raut_2}{p}{q}$ follows.

Assume $\probmass{\Raut_1} +\probmass{\Raut_2}  \leq 1$. Then
\begin{align*}
\probmass{\probchoice{\Raut_1}{\Raut_2}{p}{q}} &= \sum_{\sigma\in\N^V} p\cdot\semanticsaut{\Raut_1}(\sigma) + q\cdot\sum_{\sigma\in\N^V}\semanticsaut{\Raut_2}(\sigma') \\
&= p\probmass{\Raut_1} + q\probmass{\Raut_2}  \\
&\leq\probmass{\Raut_1} + \probmass{\Raut_2} \\
&\leq 1
\end{align*}
Assume $p + q \leq 1$. Then 
\begin{align*}
\probmass{\probchoice{\Raut_1}{\Raut_2}{p}{q}} &= \sum_{\sigma\in\N^V} p\cdot\semanticsaut{\Raut_1}(\sigma) + \sum_{\sigma\in\N^V}q\cdot\semanticsaut{\Raut_2}(\sigma') \\
&= p\probmass{\Raut_1} + q\probmass{\Raut_2}  \\
&\leq p + q\\
&\leq 1
\end{align*}

\noindent\textit{Case 5}: By construction, the initial and final weights of $\Raut\times\Baut_\guard$ are the multiplication of the initial and final weights of $\Raut$ and $\Baut_\guard$. Thus, normalization follows from the normalization of $\Raut$ and $\Baut_\guard$.

From \Cref{thm:filter} it follows that $\sum_{\sigma\in\N^V} \semanticsaut{\Raut\times\Baut_\guard}(\sigma) \leq \sum_{\sigma\in\N^V}\semanticsaut{\Raut}(\sigma)$ and thus $\Raut\times\Baut_\guard$ is also PGA.

\noindent\textit{Case 6}: Normalization follows from Cases 4 and 5.

We can see that 

\begin{align*}
    \sum_{\sigma\in\N^V} \semanticsaut{\Raut}(\sigma) &= \sum_{\sigma\in\N^V, \sigma(X)= 0} \semanticsaut{\Raut}(\sigma) + \sum_{\sigma\in\N^V, \sigma(X)\neq 0} \semanticsaut{\Raut}(\sigma) \\
    &= \sum_{\sigma\in\N^V} \semanticsaut{\Raut[X/0]}(\sigma) + \sum_{\sigma\in\N^V} \semanticsaut{\Raut\times\Baut_{X>0}}(\sigma)
\end{align*}
Since $\Raut$ is PGA, so is $\Raut^{\decr{X}}$ from Cases 1 and 5 (as the substitution by 1 does not change the probability mass).
\end{proof}
Thus we can see that every operation we introduce preserves the normalization (and PGA-property, if certain requirements are met). We will see now that these requirements are \emph{always} met with the semantics we propose. 

\endofunction*

\begin{proof}
We prove this statement by structural induction on $\prog$. Assume that the input automaton $\Raut$ is a normalized PGA. 

\noindent\textit{Case }$\assign{X}{0}$: From \Cref{lemma:properties-constructions}(1).

\noindent\textit{Case }$\incr{X}{n}$: From \Cref{lemma:properties-constructions}(3) since $\Raut_{\dirac{n}{X}}$ is a normalized PGA.

\noindent\textit{Case }$\incr{X}{D}$: From \Cref{lemma:properties-constructions}(3) since $\Raut_{D_X}$ is a normalized PGA.

\noindent\textit{Case }$\incr{X}{Y}$: From \Cref{lemma:properties-constructions}(2) since $\inlineTransTwo{Y}{X}$ is a normalized PGA.

\noindent\textit{Case }$\incr{X}{\iid{D}{Y}}$: From \Cref{lemma:properties-constructions}(2,3) since $\inlineTrans{Y}$ and $\Raut_{D_X}$ are normalized PGA.

\noindent\textit{Case }$\decr{X}$: From \Cref{lemma:properties-constructions}(6).

\noindent\textit{Case }$\observe{\guard}$: From \Cref{lemma:properties-constructions}(5).

\noindent\textit{Induction Hypothesis (IH)}: We assume that for any program $\prog$, $\semantics{\prog}$ is an endofunction on the set of normalized PGA over variables $V$.

\noindent\textit{Case }$\coinflip{\prog_1}{p}{\prog_2}$: From \Cref{lemma:properties-constructions}(4) since $p + (1-p) = 1\leq 1$ and IH.

\noindent\textit{Case }$\ite{\guard}{\prog_1}{\prog_2}$: From \Cref{lemma:properties-constructions}(4,5), \Cref{thm:filter} since \\$\sum_{\sigma\in\N^V}\semanticsaut{\Raut\times\Baut_\guard}(\sigma) + \sum_{\sigma\in\N^V}\semanticsaut{\Raut\times\Baut_{\neg\guard}}(\sigma) = \sum_{\sigma\in\N^V}\semanticsaut{\Raut}(\sigma) \leq 1$ and IH.

\noindent\textit{Case }$\prog_1\fatsemi\prog_2$: From IH.

\end{proof}

\subsection{Proof of \Cref{thm:complexity-semantics}}
To show \Cref{thm:complexity-semantics} we first examine the complexity of every construction we introduce. For this, we define by $\sizeaut{\Raut}_X$ for $X\in V$ the amount of $X$-transitions in $\Raut$. Similarly, by $|M|$ we denote the number of nonzero entries in a vector/matrix $M$. Let $I(\Raut)$ and $F(\Raut)$ denote the initial and final vector of $\Raut$, respectively. In \Cref{tab:complexity-base-cases} we provide the complexity of the distribution and guard automata. In \Cref{tab:complexity-constructions} we see the complexity of all constructions introduced in \Cref{sec:semantics}. Lastly, the complexity of the semantics defined in \Cref{tab:semantics} is given in \Cref{tab:complexity-semantics}. We also provide upper bounds for each of the complexities.

\begin{table}[t]
    \caption{Complexity of distribution and guard automata for $n, m\in\N$ with $m>n, X\in V, p\in [0,1]$ and $\guard, \guard_1,\guard_2$ Boolean guards.}
    \label{tab:complexity-base-cases}
    \centering
    \medskip
    \setlength{\tabcolsep}{15pt}
    \rowcolors{2}{gray!7}{white}
    \begin{tabular}{l l l}
        \toprule
        Automaton & $\sizeaut{\cdot}$\\
        \midrule
        $\Raut_{\dirac{n}{X}}$ & $\max\set{n-1, 1}$\\
        $\Raut_{\bern{p}{X}}$ & $1$ \\
        $\Raut_{\geom{p}{X}}$ & $1$ \\
        $\Raut_{\negbinomial{m}{p}{X}}$ & $2m - 1$ \\
        $\Raut_{\unif{m}{X}}$ & $\max\set{m-1, 1}$\\
        $\Baut_{X<0}$ &$1$ \\
        $\Baut_{X<n}$ & $\max\set{n, 1}$\\
        $\Baut_{X\equiv_m n}$ & $m$\\
        $\Baut_{\neg\guard}$ & $\sizeaut{\Baut_\guard}$ \\
        $\Baut_{\guard_1\land\guard_2}$ &$\sizeaut{\Baut_{\guard_1}}\sizeaut{\Baut_{\guard_2}}$ \\
        \bottomrule
    \end{tabular}
\end{table}

\begin{table}[t]
    \caption{Complexity of automata constructions. $\Raut, \Raut_1,\Raut_2$ are arbitrary normalized PGA, $p,q\in [0,1]$ and $X,Y \in V$.}
    \label{tab:complexity-constructions}
    \centering
    \medskip
    \setlength{\tabcolsep}{15pt}
    \rowcolors{2}{gray!7}{white}
    \begin{tabular}{l l l}
        \toprule
        Construction $c$ & $\sizeaut{\cdot}$ \\
        \midrule
        $\Raut[X/0]$ & $\max\set{\sizeaut{\Raut} - \sizeaut{\Raut}_X, 1}$\\
        $\Raut[X/1]$ & $\sizeaut{\Raut}$\\
        $\Raut_1\concat\Raut_2$ & $\sizeaut{\Raut_1} + \sizeaut{\Raut_1} + |F(\Raut_1)||I(\Raut_2)|$\\
        $\probchoice{\Raut_1}{\Raut_2}{p}{q}$ & $\sizeaut{\Raut_1} + \sizeaut{\Raut_2}$\\
        $\Raut_1[Y/\Raut_2]$ & $\sizeaut{\Raut_1} + \sizeaut{\Raut_1}_Y(\sizeaut{\Raut_2} + |F(\Raut_2)|)$\\
        $\Raut_1 \times \Raut_2$ &  $\sizeaut{\Raut_1}\sizeaut{\Raut_2}$\\
        $\Raut^{\decr{X}}$ & $3\sizeaut{\Raut} - \sizeaut{\Raut}_X$ \\
        \bottomrule
    \end{tabular}
\end{table}

\begin{table}[t]
    \caption{Complexity of semantics and upper bound on the complexity. $\Raut$ is an arbitrary normalized PGA, $X,Y \in V$, $n \in\N$, $p \in [0,1]$, $\guard$ a Boolean guard, $D$ a PGA-definable distribution (with corresponding automaton $\Raut_D$), $N = \max\set{m\colon m \text{ constant in } \prog}$ the largest constant in $\prog$ and $\prog_1,\prog_2$ \redip programs.}
    \label{tab:complexity-semantics}
    \centering
    \smallskip
    \setlength{\tabcolsep}{4pt}
    \rowcolors{2}{gray!7}{white}
    \resizebox{\textwidth}{!}{
        \begin{tabular}{l l l}
            \toprule
            $\prog$ & $\sizeaut{\semantics{\prog}(\Raut)}$ & $\geq \sizeaut{\semantics{\prog}(\Raut)}$\\
            \midrule
            $\assign{X}{0}$ &$\max\set{\sizeaut{\Raut} - \sizeaut{\Raut}_X, 1}$&$\sizeaut{\Raut}$ \\
            $\incr{X}{n}$ &$\sizeaut{\Raut}+|F(\Raut)|+n$&$2\sizeaut{\Raut}+ N$ \\
            $\incr{X}{D}$ &$\sizeaut{\Raut}+\sizeaut{\Raut_D}+|F(\Raut)||I(\Raut_D)|$&$\sizeaut{\Raut} + \sizeaut{\Raut_D} +\sizeaut{\Raut}\sizeaut{\Raut_D}$ \\
            $\incr{X}{Y}$ & $\sizeaut{\Raut} + 3\sizeaut{\Raut}_Y$& $4\sizeaut{\Raut}$ \\
            $\incr{X}{\iid{D}{Y}}$ &$\sizeaut{\Raut} + \sizeaut{\Raut}_Y(\sizeaut{\Raut_D} + |F(\Raut_D)| + |I(\Raut_D)|)$&$\sizeaut{\Raut} + 3\sizeaut{\Raut}\sizeaut{\Raut_D}$ \\
            $\decr{X}$ &$3\sizeaut{\Raut} - \sizeaut{\Raut}_X$& $3\sizeaut{\Raut}$\\
            $\observe{\guard}$ &$\sizeaut{\Raut}\sizeaut{\Baut_\guard}$& $\sizeaut{\Raut}(N+1)^{\sizebool{\guard}}$ (\Cref{lemma:complexity-auxiliary})\\
            $\coinflip{\prog_1}{p}{\prog_2}$ & $\sizeaut{\semantics{\prog_1}(\Raut)} +\sizeaut{\semantics{\prog_2}(\Raut)}$& $2\max_{\prog}\{\sizeaut{\semantics{\prog}(\Raut)}\}$ \\
            $\ite{\guard}{\prog_1}{\prog_2}$ &$\sizeaut{\semantics{\prog_1}(\Raut\times\Baut_{\guard})} + \sizeaut{\semantics{\prog_2}(\Raut\times\Baut_{\neg\guard})}$& $2\max_{\prog}\{\sizeaut{
            \semantics{\prog}(\Raut\times\Baut_\guard)
            } \}$  \\
            $\prog_1\fatsemi\prog_2$ &$\sizeaut{\semantics{\prog_2}(
            \semantics{\prog_1}(\Raut)
            )}$& $\max_{\prog_1,\prog_2}\{\sizeaut{\semantics{\prog_2}(
            \semantics{\prog_1}(\Raut)
            )}\}$\\
            \bottomrule
        \end{tabular}
    }
\end{table}

One crucial observation is that the complexity of individual guard automata is bounded by the size of the largest constant $n$ in the program. Therefore, we can give a general over-approximation for the size of \emph{any} guard automaton depending on the number of conjunctions which the guard automaton contains. Note that we could also prove a similar lemma for supported distribution automata. However, since the language also supports custom distributions (as long as they are PGA-definable and have a probability mass of 1) we use the size of the largest distribution automaton as another parameter.

\begin{lemma}[Complexity of Auxiliary Automata]\label{lemma:complexity-auxiliary}
Let $\prog$ be a \redip\ program and $n = \max\set{m\colon\ m \text{ constant in }\prog}$ be the largest number in $\prog$. Then for any guard $\guard$, 
$|\Baut_\guard| \leq (n+1)^{{\sizebool{\guard}}}$.
\end{lemma}

\begin{proof}
We prove this statement by induction on $\sizebool{\guard}$. Let $\guard$ s.t.\ $\sizebool{\guard} = 1$. From \Cref{tab:complexity-base-cases} we can see that $\sizeaut{\Baut_\guard}\leq \max\set{1,n,m} \leq (n+1) = (n+1)^{\sizebool{\guard}}$ since $n$ is the biggest number in $\prog$ by assumption (also accounting for $n=0$).

\noindent\textit{Induction Hypothesis (IH)}: We assume that for any fixed length $\ell\in\N$ with $\sizebool{\guard} = \ell$ we have $\sizeaut{\Baut_\guard}\leq n^{\sizebool{\guard}}$.

\noindent Let $\guard$ s.t.\ $\sizebool{\guard} = \ell+1$. Then $\guard = \guard_1 \land \guard_\ell$ with $\sizebool{\guard_1} = 1$ and thus $\sizebool{\guard_\ell} = \ell$. We have:
\begin{align*}
\sizeaut{\Baut_{\guard}} &= \sizeaut{\Baut_{\guard_1\land\guard_\ell}} \\
&= \sizeaut{\Baut_{\guard_1}}\sizeaut{\Baut_{\guard_\ell}} \\
&\leq  \sizeaut{\Baut_{\guard_1}}(n+1)^{\sizebool{\guard_\ell}}\quad \text{(IH)} \\
&\leq (n+1)(n+1)^{\sizebool{\guard_\ell}} \\
&= (n+1)^{\sizebool{\guard_\ell} + 1} \\
&= (n+1)^{\sizebool{\guard}}
\end{align*}
\end{proof}
With this lemma we can now characterize the size of the automaton describing the posterior distribution.
\begin{lemma}[Complexity of Base Statements]\label{lemma:complexity-base}
For any \redip program $\prog$ with $\sizeprog{\prog} = 1$ and any PGA $\Raut$ we have
\[
	\sizeaut{\semantics{\prog}(\Raut)} \in\bigO{\sizeaut{\Raut}\cdot\mathcal{D}\cdot (n+1)^{\eta}}
\]
where $\mathcal{D} = \max\set{\sizeaut{\Raut_D} \colon\ \Raut_D \text{ distribution automaton in }\prog}$ is the size of the largest distribution automaton in $\prog$, $n = \max\set{m\colon\ m \text{ constant in }\prog}$ is the largest constant in $\prog$ and $\eta=\max\set{\sizebool{\guard}\colon\ \guard \text{ guard in }\prog}$ is the size of the largest guard in $\prog$ (with the convention that $\max\emptyset = 1$).
\end{lemma}

\begin{proof}
Let $\prog$ be a \redip\ program s.t.\ $\sizeprog{\prog} = 1$, $\Raut$ a PGA, $\mathcal{D}$ the size of the largest distribution automaton in $\prog$, $n = \max\set{m\colon\ m \text{ constant in }\prog}$ the largest constant in $\prog$ and $\eta = \max\set{\sizebool{\guard}\colon\ \guard \text{ guard in }\prog}$ the size of the largest guard in $\prog$. From \Cref{tab:complexity-semantics} we can see that 
\begin{align*}
\sizeaut{\semantics{\prog}(\Raut)} &\leq \max\{\sizeaut{\Raut}, 2\sizeaut{\Raut}+n, \sizeaut{\Raut} + \sizeaut{\Raut_D} + \sizeaut{\Raut}\sizeaut{\Raut_D}, 4\sizeaut{\Raut},\\
&\qquad \sizeaut{\Raut} + 3\sizeaut{\Raut}\sizeaut{\Raut_D}+3\sizeaut{\Raut}+\sizeaut{\Raut}(n+1)^{\sizebool{\guard}}\} \\ 
&= \max\{2\sizeaut{\Raut} + n,
\sizeaut{\Raut} + \sizeaut{\Raut_D} + \sizeaut{\Raut}\sizeaut{\Raut_D},
4\sizeaut{\Raut}, \sizeaut{\Raut} + 3\sizeaut{\Raut}\sizeaut{\Raut_D}, \sizeaut{\Raut}(n+1)^{\sizebool{\guard}}\}
\end{align*}
From this we have
\begin{align*}
\sizeaut{\semantics{\prog}(\Raut)} &\in \bigO{
\sizeaut{\Raut}\cdot\mathcal{D} + \sizeaut{\Raut}\cdot (n+1)^{\sizebool{\guard}} + n
} \\
&= \bigO{
\sizeaut{\Raut}\cdot\mathcal{D} + \sizeaut{\Raut}\cdot (n+1)^{\eta}
} &&(\sizeaut{\Raut},\eta\geq 1)\\
&\subseteq \bigO{\sizeaut{\Raut}\cdot\mathcal{D}\cdot (n+1)^{\eta}} && (\mathcal{D} \geq 1)
 \\
\end{align*}
Note that since by convention $\max\emptyset = 1$ this analysis also covers programs without constants (or with $n = 0$) or sampling. 
\end{proof}

\complexitysemantics*
\begin{proof}
Let $\Raut$ be a PGA, $\mathcal{D} = \max\set{\sizeaut{\Raut_D}\colon\ \Raut_D\text{ distribution automaton in }\prog}$ the size of largest distribution automaton in $\prog$, $n = \max\set{m\colon\ m \text{ constant in }\prog}$ the largest constant in $\prog$ and $\eta  =\max\set{\sizebool{\guard}\colon\ \guard \text{ guard in }\prog}$ the size of the largest guard in $\prog$. We prove this claim by induction on the program size $\sizeprog{\prog}$.

Let $\prog$ s.t.\ $\sizeprog{\prog} = 1$. Then from \Cref{lemma:complexity-base} we have 
\begin{align*}
	\sizeaut{\semantics{\prog}(\Raut)} \in\bigO{\sizeaut{\Raut}\cdot \mathcal{D}\cdot (n+1)^{\eta}}
	&= \bigO{\sizeaut{\Raut}\cdot \mathcal{D}^{\sizeprog{\prog}}\cdot (n+1)^{\sizeprog{\prog}\eta}}
\end{align*}

\noindent\textit{Induction Hypothesis (IH)}: We assume that for any öGA $\Raut$ and $\prog$ with $\sizeprog{\prog} = \ell$ we have $\sizeaut{\semantics{\prog}(\Raut)} \in\bigO{\sizeaut{\Raut}\cdot\mathcal{D}^{\sizeprog{\prog}}\cdot (n+1)^{\sizeprog{\prog}\eta}}$.

\noindent First note that if- and probabilistic choice statements increase the size of the automaton less than the sequential composition, hence we need not to consider them here. Let $\prog=\prog_\ell\fatsemi\prog_1$ with $\sizeprog{\prog_1} = 1$. Then
\[
\sizeprog{\semantics{\prog_\ell\fatsemi\prog_1}(\Raut)} = \sizeprog{(\semantics{\prog_1}\underbrace{\semantics{\prog_\ell}(\Raut)}_{\in \bigO{\sizeaut{\Raut}\cdot\mathcal{D}^{\sizeprog{\prog_\ell}}\cdot (n+1)^{\sizeprog{\prog_\ell}\eta}}})}
\]
Hence 
\begin{align*}
\sizeprog{\semantics{\prog_\ell\fatsemi\prog_1}(\Raut)} &\in \bigO{(\sizeaut{\Raut}\cdot \mathcal{D}^{\sizeprog{\prog_\ell}}\cdot (n+1)^{\sizeprog{\prog_\ell}\eta})\cdot \mathcal{D}\cdot (n+1)^\eta}\\
&= \bigO{\sizeaut{\Raut}\cdot\mathcal{D}^{\sizeprog{\prog_\ell} +1}\cdot (n+1)^{(\sizeprog{\prog_\ell}+1)\eta}} \\
&= \bigO{\sizeaut{\Raut}\cdot\mathcal{D}^{\sizeprog{\prog}}\cdot (n+1)^{\sizeprog{\prog}\eta}}
\end{align*}
\end{proof}

\begin{figure}[t]
    \centering 
    \addtolength{\jot}{1em} 
    \begin{gather*}
        \sosrule{Assgn}{}{\opstate{\assign{X}{E}}{\val} ~\rightarrow~ \opstate{\done}{\val[X \gets E(\sigma)]}} \\
        \sosrule{Done}{}{\opstate{\done}{\val}~\rightarrow~\opstate{\done}{\sigma}}\qquad \sosrule{Error}{}{\errorstate ~\rightarrow~\errorstate} \\
        \sosrule{Obs-T}{\val\models\guard}{\opstate{\observe{\guard}}{\val} ~\rightarrow~ \opstate{\done}{\sigma}} \qquad
        \ruletag{Obs-f} \frac{\val\not\models\guard}{\opstate{\observe{\guard}}{\val}~\rightarrow~\errorstate }\\
        \sosrule{Seq-1}{\opstate{\prog_1}{\val}~\xrightarrow{p}~\opstate{\prog_1'}{\val'}}{\opstate{\prog_1\fatsemi\prog_2}{\sigma}~\xrightarrow{p}~\opstate{\prog_1'\fatsemi\prog_2}{\val'}} \qquad
        \sosrule{Seq-2}{\opstate{\prog_1}{\val}~\xrightarrow{p}~\opstate{\done}{\val'}}{\opstate{\prog_1\fatsemi\prog_2}{\val}~\xrightarrow{p}~\opstate{\prog_2}{\val'}} \\
        \sosrule{Seq-3}{\opstate{\prog_1}{\sigma}~\rightarrow~\errorstate}{\opstate{\prog_1\fatsemi\prog_2}{\val}~\rightarrow~\errorstate} \qquad
        \sosrule{Sample}{D \in Distr(\N) \qquad D(n) = p}{\opstate{\incr{X}{D}}{\val} ~\xrightarrow{p}~ \opstate{\done}{\val[X \gets \val(X) + n]}} \\
        \sosrule{Choice-L}{}{\opstate{\coinflip{\prog_1}{p}{\prog_2}}{ \val} ~\xrightarrow{p}~\opstate{\prog_1}{\val}} \\
        \sosrule{Choice-R}{}{\opstate{\coinflip{\prog_1}{p}{\prog_2}}{ \val} ~\xrightarrow{1-p}~\opstate{\prog_2}{\val}} \\
        \sosrule{If-T}{\val\models\guard}{\opstate{\ite{\guard}{\prog_1}{\prog_2}}{\val}~\rightarrow~ \opstate{\prog_1}{\val}} \\
        \sosrule{If-F}{\val\not\models\guard}{\opstate{\ite{\guard}{\prog_1}{\prog_2}}{\val}~\rightarrow~ \opstate{\prog_2}{\val}} \\
    \end{gather*}
    \caption{Construction Rules for the Operational Markov Chain Semantics \cite{DBLP:journals/toplas/OlmedoGJKKM18}. $\sigma[X\gets n]$ indicates a state $\sigma'$ where $\sigma'(X) = n$ and $\sigma'(Y) = \sigma(Y)$ for all other variables $Y \in V \setminus\set{X}$.}
\end{figure}

\subsection{Proof of \Cref{thm:opequiv}}\label{sec:op-equiv}

To prove the soundness of the program semantics, we first need a more precise characterization of how the syntactic statements change the semantics of the automaton.  
\begin{definition}[Path]\label{def:paths}%
    We call a sequence of states $P= s_1,\hdots, s_n$ a path on an automaton $\Raut = (Q,M,I,F)$ with $\set{s_1,\hdots,s_n}\subseteq Q$. The weight of a path is defined by\[
    \semanticspath{s_1,\hdots,s_n} = \prod_{i=1}^{n-1} M_{s_{i},s_{i+1}}
    \]
    and $\semanticsapath{s_1,\hdots,s_n} = I_{s_1}\semanticspath{s_1,\hdots,s_n}F_{s_n}$.
    We call a path accepting if $\semanticsapath{s_1,\hdots,s_n}\neq 0$. The length of a path $\lengthpath{P}$ is the length of the sequence. We denote the set of all accepting paths on $\Raut$ by $\APaths{\Raut}$.
\end{definition}

\begin{lemma}[{\cite[Chap.\ 4, Lem.\ 3.2]{weighted-automata}}]\label{lemma:power-transition}%
    For $\Raut = (Q,M,I,F)$ and every $n \in\N$ it holds that 
    \[
M^n_{s,t} = \sum_{P=s,\hdots,t;\lengthpath{P} = n} \semanticspath{P}\qquad \textrm{for all }s,t\in Q
    \]
\end{lemma}

\begin{lemma}
    For $\Raut = (Q,M,I,F)$ we have $\semanticsaut{\Raut} = \sum_{P\in\APaths{\Raut}} \semanticsapath{P}$.
\end{lemma}

\begin{proof}
    \begin{align*}
          \semanticsaut{\Raut} &= IM^*F &\\
        &= \sum_{s,t \in Q} I_s\bigl(M^*\bigr)_{s,t}F_t &\\
        &= \sum_{s,t \in Q} I_s\Bigl(\sum_{i\geq 0} M^i\Bigr)_{s,t}F_t &\\
        &= \sum_{s,t \in Q} I_s\Bigl(\sum_{i\geq 0} M_{s,t}^i\Bigr)F_t& \\
        &= \sum_{s,t\in Q} I_s \Bigl(\sum_{i\geq 0} \sum_{P=s,\hdots,t; |P| = i} \semanticspath{P}\Bigr) F_t && \text{(\Cref{lemma:power-transition})}
    \end{align*}
    We can see that we enumerate all possible paths starting with increasing length and multiplied by the initial and final weight of $s$ and $t$ respectively. By distributivity, we can see that this is equal to the sum of all accepting paths by \Cref{def:paths}. 
\end{proof}
We now examine how the semantics of the individual automata constructions change the semantics of the underlying PGF represented by the input PGA.
\begin{lemma}\label{lemma:semantics-constructions}%
       For $i\in\set{0,1}, p,q\in [0,1], X,Y \in V$ and $\Raut, \Raut_1,\Raut_2$ normalized PGA we have:
       \begin{enumerate}
        \item $\semanticsaut{\Raut[X/i]} = \semanticsaut{\Raut}[X/i]$
        \item $\semanticsaut{\Raut_1[Y/\Raut_2]} = \semanticsaut{\Raut_1}[Y/\semanticsaut{\Raut_2}]$
        \item $\semanticsaut{\probchoice{\Raut_1}{\Raut_2}{p}{q}} = p\semanticsaut{\Raut_1} + q\semanticsaut{\Raut_2}$ 
        \item $\semanticsaut{\Raut_1\concat\Raut_2} = \semanticsaut{\Raut_1}\cdot\semanticsaut{\Raut_2}$
        \item $\semanticsaut{\Raut^{\decr{X}}} = \semanticsaut{\Raut\times\Baut_{X>0}}X^{-1} + \semanticsaut{\Raut}[X/0]$
    \end{enumerate}
\end{lemma}

\begin{proof}
Let $p,q\in [0,1], X,Y \in V$ and $\Raut, \Raut_1,\Raut_2$ normalized PGA.

\noindent\textit{Case 1}: Let $P \in\APaths{\Raut}$. For $i = 0$, we can see that that the corresponding path $P'$ is only accepting on $\Raut[X/0]$ if no transition on $P$ contains any $X$-transition. Hence
\begin{align*}
        \semanticsaut{\Raut[X/0]} &= \sum_{P'\in\APaths{\Raut[X/0]}} \semanticsapath{P'} &\\
    &= \sum_{P\in\APaths{\Raut}, P \textrm{ has no }X \textrm{-transition}} \semanticsapath{P} &\\
    &= \semanticsaut{\Raut}[X/0] && \text{(\cite[Lem.\ D.1]{redip})}
\end{align*}
We can also see that the corresponding path $P'$ is accepting on $\APaths{\Raut[X/1]}$ with $\semanticsapath{P} = \semanticsapath{P'}X^{n_P}$ where $n_P$ is the number of $X$-transitions of $P$. Hence

\begin{align*}
      \semanticsaut{\Raut}[X/1] &= \left(\sum_{P \in \APaths{\Raut}} \semanticsapath{P}\right)[X/1] & \\
  &= \left(\sum_{P'\in\APaths{\Raut[X/1]}}\semanticsapath{P'}X^{n_P} \right)[X/1] \\
  &= \sum_{P' \in\APaths{\Raut[X/1]}}\semanticsapath{P'} & \\
  &= \semanticsaut{\Raut[X/1]} 
\end{align*}

\noindent\textit{Case 2}: Let $P = s_1,\hdots, s_n \in \APaths{ \Raut_1}$ and $P'=s'_1,\hdots,s'_{n'} \in\APaths{\Raut_2}$. We can see that if $P$ contains at least one $Y$-transition, we can take an arbitrary accepting path from $ \Raut_2$ and end up back in $\Raut_1$ (by construction). Assume $P$ has $m$ $Y$-transition, indexed by $i_1,\hdots,i_m$ with $i_1 < \cdots < i_m$. Hence the corresponding new path (assuming we always choose $P'$) is given by 
\[
 P[Y/P'] := s_0,\hdots,s_{i_1},\hdots,s'_1,\hdots,s'_{n'},s_{i+1}+1,\hdots,s_{i_m},s'_1,\hdots,s'_{n'},s_{i_m}+1,\hdots,s_n
\]
We can see that $\semanticsapath{P[Y/P']} = \semanticsapath{P}[Y/1]\prod_{m} \semanticsapath{P'}$. 
Since $P'$ is arbitrary and can be chosen independently each time, we can now examine the values of the accepting paths that are \say{obtained} from $P$: 
\begin{align*}
&\quad \sum_{P_1\in\APaths{\Raut_2}}\cdots\sum_{P_m \in\APaths{\Raut_2}} \semanticsapath{P}[Y/1]\prod_{i=1}^m \semanticsapath{P_i}  \\
=&\quad\semanticsapath{P}[Y/1]\Bigl(\sum_{P_1\in\APaths{\Raut_2}}\cdots\sum_{P_m \in\APaths{\Raut_2}} \prod_{i=1}^m \semanticsapath{P_i}\Bigr) \\
=&\quad \semanticsapath{P}\Bigl[Y/\sum_{P\in\APaths{\Raut_2}}\semanticsapath{P}\Bigr]\\
=&\quad \semanticsapath{P}[Y/\semanticsaut{\Raut_2}]
\end{align*}
Thus we have \begin{align*}
    \semanticsaut{\Raut[Y/\Raut_2]} &= \sum_{P\in\APaths{\Raut[Y/\Raut_2]}} \semanticsapath{P} \\
    &= \sum_{P'\in\APaths{\Raut}} \semanticsapath{P}[Y/\semanticsaut{\Raut_2}] \\
    &= \semanticsaut{\Raut_1}[Y/\semanticsaut{\Raut_2}]
\end{align*}

\noindent\textit{Case 3}: Let $P_1\in\APaths{\Raut_1}$ and $P_2 \in\APaths{\Raut_2}$. We can see that the corresponding paths $P_1', P_2' \in\APaths{\probchoice{\Raut_1}{\Raut_2}{p}{q}}$ with $\semanticsapath{P_1} = p\semanticsapath{P_1'}$ and $\semanticsapath{P_2} = q\semanticsaut{P'_2}$. Hence
\begin{align*}
    \semanticsaut{\probchoice{\Raut_1}{\Raut_2}{p}{q}} &= \sum_{P \in\APaths{\probchoice{\Raut_1}{\Raut_2}{p}{q}}} \semanticsapath{P} \\
    &= \sum_{P_1 \in\APaths{\Raut_1}} p \semanticsapath{P_1} ~+~ \sum_{P_2\in\APaths{\Raut_2}} q\semanticsapath{P_2} \\
    &= p\cdot\sum_{P_1 \in\APaths{\Raut_1}} \semanticsapath{P_1} ~+~ q\cdot\sum_{P_2 \in\APaths{\Raut_2}} \semanticsapath{P_2}  \\
    &= p\semanticsaut{\Raut_1} + q\semanticsaut{\Raut_2}
\end{align*}

\noindent\textit{Case 4}: We can see that for every $P_1 \in\APaths{\Raut_1}, P_2 \in\APaths{\Raut_2}$ there exists a path $P \in\APaths{\Raut_1\concat\Raut_2}$ with $\semanticsapath{P} = \semanticsapath{P_1}\cdot\semanticsapath{P_2}$. Hence
\begin{align*}
    \semanticsaut{\Raut_1\concat\Raut_2} &= \sum_{P\in\APaths{\Raut_1\concat \Raut_2}} \semanticsapath{P} \\
    &= \sum_{P_1\in\APaths{\Raut_1}, P_2 \in\APaths{\Raut_2}} \semanticsapath{P_1}\cdot \semanticsapath{P_2} \\
    &= \left(\sum_{P_1\in\APaths{\Raut_1}} \semanticsapath{P_1}\right)\left(\sum_{P_2\in\APaths{\Raut_2}}\semanticsapath{P_2}\right) \\
    &= \semanticsaut{\Raut_1}\cdot\semanticsaut{\Raut_2}
\end{align*}

\noindent\textit{Case 5}: Let $P \in\APaths{\Raut}$. Then there are two possibilities: 
\begin{enumerate}
    \item $P$ contains at least one $X$-transition
    \item $P$ contains zero $X$-transitions
\end{enumerate}
We can see that we \say{partition} the set of all accepting paths using the product with $\Baut_{X>0}$, i.e.\ only accepting paths of Case 1 are considered included in $\Raut\times\Baut_{X>0}$ and all other accepting paths are considered in $\Raut[X/0]$. Assume $P\in\APaths{\Raut\times\Baut_{X>0}}$ with $P = s_0,\hdots,s_n$. Let $i_1,\hdots,i_m \in\set{1,\hdots, n-1}$ such that $\supp{M_{i_j,i_j + 1}} \supseteq\set{X}$ (for all $j \in \set{1,\hdots, m}$), i.e.\ the indices of all outgoing $X$-transitions, with $i_1 < \cdots < i_m$. By construction of $\Raut^{\decr{X}}$ we substitute $X$ in the \emph{first} $X$-transition (i.e.\ the transition $M_{i_1,i_1 +1}$) by 1 (resulting in $P'$). We therefore have $m -1$ $X$-transitions remaining.  Thus $\semanticsapath{P} = \semanticsapath{P'}X^{-1}$. All accepting paths satisfying $X=0$ are collected in $\Raut[X/0]$ and not modified further. Hence we have $\semanticsaut{\Raut^{\decr{X}}} = \semanticsaut{\Raut\times\Baut_{X>0}}X^{-1} + \semanticsaut{\Raut}[X/0]$.
\end{proof}

We now examine the reachability probabilities in the Markov chain as defined in  \Cref{def:operationalMarkovChain}. 
\begin{lemma}\label{thm:non-norm-equiv}%
Let $\Raut$ be a normalized PGA. Then for any \redip program $\prog$ (without \texttt{iid}-statements) and $\sigma'\in\N^V$:
\begin{enumerate}
\item $\Pr^{\markovchain{\Raut}{\prog}}(\Diamond\langle \downarrow, \sigma'\rangle) = \semanticsaut{\semantics{\prog}(\Raut)}(\sigma')$
\item $\Pr^{\markovchain{\Raut}{\prog}}(\Diamond\errorstate) = \probmass{\Raut} -  \probmass{\semantics{\prog}(\Raut)}$
\end{enumerate}
\end{lemma}

\begin{proof}
We prove this claim by structural induction on $\prog$. Let $\sigma\in\N^V$ with $\semanticsaut{\Raut}(\sigma) = p$.

\noindent\textit{Case }$\prog = \texttt{x := E}$:
\begin{center}
\begin{tikzpicture}[node distance=3.5cm]
\node (i) {$\opstate{\mathtt{x += D}}{\val}$};
\node[phantom] (init) [left=0.5cm of i] {};

\node[right of =i] (term) {$\opstate{\done}{
\val[X\gets E(\val)]}$};

\path[->] (init) edge node[above] {$p$} (i);
\path[->] (i) edge (term);
\end{tikzpicture}
\end{center}
We can see that 
\[
\markovprob{\mathtt{x := E}}{\Diamond\langle\downarrow,\sigma'\rangle} = \begin{cases}
p & \sigma' = \sigma[X\gets E(\sigma)] \\
0 & else 
\end{cases}
\]
Case distinction over the allowed types of expressions:

\textit{E = 0}: Then 

\begin{align*}
    \semanticsaut{\semantics{\mathtt{x := 0}}(\Raut)}(\sigma') &= \semanticsaut{\Raut[X/1]}(\sigma') &\\
 &= \bigl(\semanticsaut{\Raut}[X/1]\bigr)(\sigma') && \text{(\Cref{lemma:semantics-constructions}(1))}\\
 &= \begin{cases}
\semanticsaut{\Raut}(\sigma) & \sigma' = \sigma[X\gets 0] \\
0 & else
\end{cases}    
\end{align*}

\textit{E = x + n}: Then 

\begin{align*}
    \semanticsaut{\semantics{\mathtt{x += n}}(\Raut)}(\sigma') &= \semanticsaut{\Raut\concat\Raut_{\dirac{n}{X}}}(\sigma') \\
&= \bigl(\semanticsaut{\Raut}\cdot\semanticsaut{\Raut_{\dirac{n}{X}}}\bigr)(\sigma') && \text{(\Cref{lemma:semantics-constructions}(4))} \\
&= \bigl(\semanticsaut{\Raut}\cdot X^n\bigr)(\sigma') \\
&= \begin{cases}
\semanticsaut{\Raut}(\sigma) & \sigma' = \sigma[X\gets \sigma(X) + n] \\
0 & else
\end{cases}  
\end{align*}

\textit{E = x + y}: Then
\begin{align*}
    \semanticsaut{\semantics{\mathtt{x += y}}(\Raut)}(\sigma') &= \semanticsaut{\Raut[Y/\inlineTransTwo{Y}{X}]}(\sigma') \\
&= \bigl(\semanticsaut{\Raut}[Y/\semanticsaut{\inlineTransTwo{Y}{X}}]\bigr)(\sigma') && \text{(\Cref{lemma:semantics-constructions}(2))} \\
&= \bigl(\semanticsaut{\Raut}[Y/YX]\bigr)(\sigma') \\
&= \begin{cases}
\semanticsaut{\Raut}(\sigma) & \sigma' = \sigma[X \gets \sigma(X) + \sigma(Y)] \\
0 & else
\end{cases}   
\end{align*}
\textit{E = x $\monus$ 1}, where $a \monus b = \max \set{a - b, 0}$: Then
\begin{align*}
\semanticsaut{\semantics{\decr{x}}(\Raut)}(\sigma') &= \semanticsaut{\Raut^{\decr{X}}}(\sigma') \\
&= \bigl(\semanticsaut{\Raut\times\Baut_{X>0}}X^{-1} + \semanticsaut{\Raut}[X/0]\bigr)(\sigma') & \text{(\Cref{lemma:semantics-constructions}(5))} \\
&= \begin{cases}
\semanticsaut{\Raut}(\sigma) & \sigma' = \sigma[X \gets \sigma(X) \monus 1] \\
0 & else
\end{cases}
\end{align*}
For all expressions $E$ as distinguished above, we have
\[
\markovprob{\mathtt{x := E}}{\Diamond\errorstate} = 0 = \displaystyle\probmass{\Raut} - \underbrace{\probmass{\semantics{\mathtt{x := E}}(\Raut)}}_{= \probmass{\Raut}}
\]
\noindent\textit{Case }$\prog = \texttt{x += D}$: Assume $D \in {Distr}(\N)$ and $D(n) = q$.
\begin{center}
\begin{tikzpicture}[node distance=3.5cm]
\node (i) {$\opstate{\mathtt{x += D}}{\val}$ };
\node[phantom] (init) [left=0.5cm of i] {};

\node[right of =i] (term) {$\opstate{\done}{\val[X\gets \val(X) + n]}$};

\path[->] (init) edge node[above] {$p$} (i);
\path[->] (i) edge node[above]{$q$}(term);
\end{tikzpicture}
\end{center}
We infer $\markovprob{\mathtt{x += D}}{\Diamond\langle \downarrow,\sigma'\rangle} = \begin{cases}
pq & \sigma' = \sigma[X \gets \sigma(X) + n] \land D(n) = q \\
0 & else
\end{cases}$

\noindent Let $\sigma_{X=n}$ such that $\sigma_{X=n}(X) = n$.
\begin{align*}
    &\semanticsaut{\semantics{\mathtt{x += D}}(\Raut)}(\sigma') \\
=\;& \semanticsaut{\Raut \concat \Raut_{D_X}} \\
=\;& \semanticsaut{\Raut}\cdot\semantics{\Raut_{D_X}} && \text{(\Cref{lemma:semantics-constructions}(4))}\\
=\;& \begin{cases}
\semanticsaut{\Raut}(\sigma)\cdot \semanticsaut{\Raut_{D_X}}(\sigma_{X=n}) & \sigma' = \sigma[X \gets \sigma(X) + n] \\
0 & else
\end{cases} \\[1.25em]
=\;& \begin{cases}
pq & \sigma' = \sigma[X \gets \sigma(X) + n] \land D(n) = q \\
0 & else
\end{cases}
\end{align*}
Since $\probmass{\Raut_{D_X}} = 1$ by requirement we have
\[
\markovprob{\mathtt{x += D}}{\Diamond\errorstate} = 0 = \displaystyle\probmass{\Raut} - \underbrace{\probmass{\semantics{\mathtt{x += D}}(\Raut)}}_{= \probmass{\Raut}}
\]
\noindent\textit{Case }$\prog = \observe{\guard}$: Assume $\sigma \models \guard$. Then
\begin{center}
\begin{tikzpicture}[node distance=3.5cm]
\node (i) {$\opstate{\mathtt{\observe{\guard}}}{\val}$ };
\node[phantom] (init) [left=0.5cm of i] {};

\node[right of =i] (term) {$\opstate{\done}{\val}$};
\path[->] (init) edge node[above] {$p$} (i);
\path[->] (i) edge(term);
\end{tikzpicture}
\end{center}
and thus \[
\markovprob{\observe{\guard}}{\Diamond\opstate{\done}{\val'}} = \begin{cases}
p & \sigma' = \sigma \\
0 & else
\end{cases}
\]
By \Cref{thm:filter} we have
\begin{align*}
\semanticsaut{\semantics{\observe{\guard}}(\Raut)}(\sigma') &= \semanticsaut{\Raut \times \Baut_{\guard}}(\sigma') \\
&= \begin{cases}
p & \sigma' = \sigma\\
0 & else
\end{cases}
\end{align*}
Assume $\sigma\not\models\guard$.
\begin{center}
\begin{tikzpicture}[node distance=3.5cm]
\node (i) {$\opstate{\mathtt{\observe{\guard}}}{\val}$ };
\node[phantom] (init) [left=0.5cm of i] {};

\node[right of =i] (term) {$\errorstate$};

\path[->] (init) edge node[above] {$p$} (i);
\path[->] (i) edge(term);

\end{tikzpicture}
\end{center}
We can see that $\markovprob{\observe{\guard}}{\Diamond\langle\downarrow,\sigma'\rangle} = 0 = \semanticsaut{\Raut\times\Baut_{\guard}}(\sigma')$. It is also clear that $\markovprob{\observe{\guard}}{\Diamond\errorstate} = \sum_{\sigma\not\models\guard} \semanticsaut{\Raut}(\sigma)$, i.e.\ the sum of initial probabilities of all states that violate the guard. Hence,

\begin{align*}
\probmass{\Raut} - \probmass{\semantics{\observe{\guard}}(\Raut)}&= \probmass{\Raut} - \probmass{\Raut\times\Baut_{\guard}} \\
&= \probmass{\Raut} - \sum_{\sigma\models\guard} \semanticsaut{\Raut}(\sigma) \\
&= \sum_{\sigma\not\models\guard} \semanticsaut{\Raut}(\sigma)
\end{align*}

\noindent\textit{Induction Hypothesis (IH)}: We assume that for any normalized PGA $\Raut$, \redip program $\prog$ (without \texttt{iid}-statements) and $\sigma' \in\N^V$ we have $\markovprob{\prog}{\Diamond\opstate{\done}{\val'}} = \semanticsaut{\semantics{\prog}(\Raut)}(\sigma')$ and $\markovprob{\prog}{\Diamond\errorstate} = \probmass{\Raut} - \probmass{\semantics{\prog}(\Raut)}$.

\noindent\textit{Case }$\prog = \coinflip{\prog_1}{q}{\prog_2}$:

Let $q \in [0,1]$.
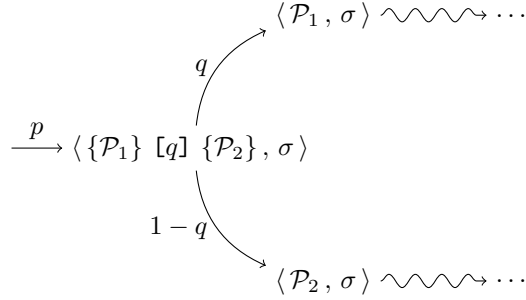
\begin{center}
\begin{tikzpicture}[node distance=2.5cm]
\node (i) {$\opstate{\coinflip{\prog_1}{q}{\prog_2}}{\val}$ };
\node[phantom] (init) [left of=i] {};

\node[above right of=i] (p1) {$\opstate{\prog_1}{\val}$};
\node[phantom, right of=p1] (php1) {$\hdots$};

\node[below right of=i] (p2) {$\opstate{\prog_2}{\val}$};

\node[phantom, right of=p2] (php2) {$\hdots$};

\path[->] (init) edge node[above] {$p$} (i);
\draw[edge, bend left] (i) to node[left]{$q$} (p1);
\draw[edge, bend right] (i) to node[left]{$1-q$} (p2);

\draw[->, decorate, decoration={snake}]
    (p1) to (php1);
    
\draw[->, decorate, decoration={snake}]
    (p2) to (php2);
\end{tikzpicture}
\end{center}
Let $\Raut_\sigma$ be a normalized PGA s.t.\ $\semanticsaut{\Raut_\sigma}(\sigma) = 1$.
We have
\begin{align*}
    &\markovprob{\coinflip{\prog_1}{q}{\prog_2}}{\Diamond\langle\downarrow,\sigma'\rangle} \\
=\;& \sum_{\sigma\in\N^V} \Bigl( \semanticsaut{\Raut}(\sigma)\cdot\bigl(q\cdot\mathrm{Pr}^{\markovchain{\Raut_\sigma}{\prog_1}}(\Diamond\langle\downarrow,\sigma'\rangle) + (1-q) \mathrm{Pr}^{\markovchain{\Raut_\sigma}{\prog_2}}(\Diamond\langle\downarrow,\sigma'\rangle)\bigr) \Bigr)\\
=\;& \probmass{\Raut}\cdot\bigl(q\cdot \semanticsaut{\semantics{\prog_1}(\Raut_\sigma)}(\sigma') + (1-q) \semanticsaut{\semantics{\prog_2}(\Raut_\sigma)}(\sigma')\bigr) \quad\quad\text{(IH)}\\
=\;& \probmass{\Raut}\cdot q\cdot \semanticsaut{\semantics{\prog_1}(\Raut_\sigma)}(\sigma') + \semanticsaut{\Raut}(\sigma)(1-q) \semanticsaut{\semantics{\prog_2}(\Raut_\sigma)}(\sigma') \\
=\;& \probmass{\Raut}\cdot q\cdot \semanticsaut{\semantics{\prog_1}(\Raut_\sigma)}(\sigma') + \sum_{\sigma\in\N^V}\semanticsaut{\Raut}(\sigma)(1-q) \semanticsaut{\semantics{\prog_2}(\Raut_\sigma)}(\sigma') \\
=\;& q\left(\sum_{\sigma\in\N^V} \semanticsaut{\Raut}(\sigma)\cdot\semanticsaut{\semantics{\prog_1}(\Raut_\sigma)}(\sigma')\right) + (1-q)\left(\sum_{\sigma\in\N^V}\semanticsaut{\Raut}(\sigma) \cdot\semanticsaut{\semantics{\prog_2}(\Raut_\sigma)}(\sigma')\right)  \\
=\;& q\cdot \semanticsaut{\semantics{\prog_1}(\Raut)}(\sigma') + (1-q)\semanticsaut{\semantics{\prog_2}(\Raut)}(\sigma') \\
=\;& \semanticsaut{\probchoice{\semantics{\prog_1}(\Raut)}{\semantics{\prog_2}(\Raut)}{q}{1-q}}(\sigma') \\
=\;& \semanticsaut{\semantics{\coinflip{\prog_1}{q}{\prog_2}}(\Raut)}(\sigma')
\end{align*}
For $\markovprob{\coinflip{\prog_1}{q}{\prog_2}}{\Diamond\errorstate}$, we have
\begin{align*}
     &\markovprob{\coinflip{\prog_1}{q}{\prog_2}}{\Diamond\errorstate} \\
=\;& q\cdot \markovprob{\mathtt{\prog_1}}{\Diamond\errorstate} + (1 - q)\markovprob{\mathtt{\prog_1}}{\Diamond\errorstate} \\
=\;& q\Bigl(\probmass{\Raut} - \probmass{\semantics{\prog_1}(\Raut)}\Bigr) + \\
&\quad (1-q)\Bigl(\probmass{\Raut} - \probmass{\semantics{\prog_2}(\Raut)}\Bigr) &\text{(IH)}\\
=\;& q\probmass{\Raut} - q\probmass{\semantics{\prog_1}(\Raut)} + \probmass{\Raut} - \\
&\quad\probmass{\semantics{\prog_2}(\Raut)} 
 - q\probmass{\Raut} + \probmass{\semantics{\prog_2}(\Raut)} \\
=\;& \probmass{\Raut} - q\probmass{\semantics{\prog_1}(\Raut)} - (1-q)\semanticsaut{\semantics{\prog_2}(\Raut)} \\
=\;& \probmass{\Raut} - \probmass{\semantics{\coinflip{\prog_1}{q}{\prog_2}}(\Raut)}
\end{align*}

\noindent\textit{Case }$\prog =
\ite{\guard}{\prog_1}{\prog_2}$: We can see that the probability of reaching a final configuration $\opstate{\done}{\sigma'}$ with the \texttt{if}-statement corresponds to the sum of reaching $\opstate{\done}{\sigma'}$ for each $\sigma \in\N^V$ with $\prog_1$ if $\sigma\models\guard$ or with $\prog_2$ if $\sigma\not\models\guard$. 
\begin{align*}
    &\markovprob{\ite{\guard}{\prog_1}{\prog_2}}{\Diamond\langle\downarrow,\sigma'\rangle} \\
=\;& \sum_{\sigma \in\N^V,\sigma\models\guard} \semanticsaut{\Raut}(\sigma)  \cdot \mathrm{Pr}^{\markovchain{\Raut_\sigma}{\prog_1}}(\Diamond\langle\downarrow,\sigma'\rangle)+ \\
&\quad \sum_{\sigma \in\N^V,\sigma\not\models\guard} \semanticsaut{\Raut}(\sigma)\cdot  \mathrm{Pr}^{\markovchain{\Raut_\sigma}{\prog_2}}(\Diamond\langle\downarrow,\sigma'\rangle) \\
=\;& \sum_{\sigma \in\N^V,\sigma\models\guard} \semanticsaut{\Raut}(\sigma) \cdot \semanticsaut{\semantics{\prog_1}(\Raut_\sigma)}(\sigma') + \\
&\quad \sum_{\sigma \in\N^V,\sigma\not\models\guard} \semanticsaut{\Raut}(\sigma)\cdot  \semanticsaut{\semantics{\prog_2}(\Raut_\sigma)}(\sigma') \quad\qquad \text{(IH)}\\ 
=\;& \semanticsaut{\semantics{\prog_1}(\Raut\times\Baut_{\guard})}(\sigma') + \semanticsaut{\semantics{\prog_2}(\Raut\times\Baut_{\neg\guard})}(\sigma') \\
=\;& \semanticsaut{\weightedunion{\semantics{\prog_1}(\Raut\times\Baut_{\guard})}{ \semantics{\prog_2}(\Raut\times\Baut_{\neg\guard})}}(\sigma')\\
=\;& \semanticsaut{\semantics{\ite{\guard}{\prog_1}{\prog_2}}(\Raut)}(\sigma')
\end{align*}
Similarly, reaching an error state $\errorstate$ corresponds to the sum of probabilities for each $\sigma\in\N^V$ of reaching an error state with $\prog_1$ if $\sigma\models\guard$ and with $\prog_2$ otherwise.
\begin{align*}
&\markovprob{\ite{\guard}{\prog_1}{\prog_2}}{\Diamond\errorstate} \\ 
=& \sum_{\sigma\models\guard}\mathrm{Pr}^{\markovchain{\Raut_\sigma}{\prog_1}}(\Diamond\errorstate) + \sum_{\sigma\not\models\guard} \mathrm{Pr}^{\markovchain{\Raut_\sigma}{\prog_2}}(\Diamond\errorstate) \\
=& \sum_{\sigma\models\guard}\semanticsaut{\Raut}(\sigma)\Bigl(
    \sum_{\sigma'\in\N^V} \semanticsaut{\Raut_{\sigma}}(\sigma') - \sum_{\sigma'\in\N^V} \semanticsaut{\semantics{\prog_1}(\Raut_{\sigma})}(\sigma') \Bigr) + \\
    &\quad \sum_{\sigma\not\models\guard}\semanticsaut{\Raut}(\sigma)\Bigl(
    \sum_{\sigma'\in\N^V} \semanticsaut{\Raut_{\sigma}}(\sigma') - \sum_{\sigma'\in\N^V} \semanticsaut{\semantics{\prog_2}(\Raut_{\sigma})}(\sigma') \Bigr) \quad \text{(IH)}\\
=& \sum_{\sigma\models\guard}\semanticsaut{\Raut}(\sigma)\Bigl(
    1 - \sum_{\sigma'\in\N^V} \semanticsaut{\semantics{\prog_1}(\Raut_{\sigma})}(\sigma') \Bigr) + \\
    &\quad \sum_{\sigma\not\models\guard}\semanticsaut{\Raut}(\sigma)\Bigl(
    1 - \sum_{\sigma'\in\N^V} \semanticsaut{\semantics{\prog_2}(\Raut_{\sigma})}(\sigma') \Bigr) \\
=& \sum_{\sigma\models\guard}\Bigl(\semanticsaut{\Raut}(\sigma) - \semanticsaut{\Raut}(\sigma)\sum_{\sigma' \in\N^V} \semanticsaut{\semantics{\prog_1}(\Raut_{\sigma})}(\sigma')  
    \Bigr) \\
    &\quad\sum_{\sigma\not\models\guard}\Bigl(\semanticsaut{\Raut}(\sigma) - \semanticsaut{\Raut}(\sigma)\sum_{\sigma' \in\N^V} \semanticsaut{\semantics{\prog_2}(\Raut_{\sigma})}(\sigma')  
    \Bigr) \\
=& \probmass{\Raut} - \Bigl(\sum_{\sigma\in\N^V,\sigma\models\guard}\semanticsaut{\Raut}(\sigma) \sum_{\sigma'\in\N^V} \semanticsaut{\semantics{\prog_1}(\Raut_\sigma)}(\sigma') + \\
&\quad \sum_{\sigma\in\N^V,\sigma\not\models\guard} \semanticsaut{\Raut}(\sigma) \sum_{\sigma'\in \N^V} \semanticsaut{\semantics{\prog_2}(\Raut_\sigma)}(\sigma')\Bigr) \\
=& \probmass{\Raut} - \Bigl(\sum_{\sigma\in\N^V} \semanticsaut{\semantics{\prog_1}(\Raut\times\Baut_\guard)}(\sigma) +\sum_{\sigma\in\N^V} \semanticsaut{\semantics{\prog_2}(\Raut\times\Baut_{\neg\guard})}(\sigma)\Bigr) \\
=& \probmass{\Raut} - \probmass{\semantics{\ite{\guard}{\prog_1}{\prog_2}}(\Raut)}  
\end{align*}

\noindent\textit{Case }$\prog = \prog_1\fatsemi\prog_2$:
First, we observe that the probability of eventually reaching $\sigma'$ from $\sigma$ with the program $\prog_1\fatsemi\prog_2$ is reaching an intermediate state $\sigma''$ after executing $\prog_1$ and then reaching $\sigma'$ from $\prog_2$ with $\sigma''$. Hence, 
\begin{align*}
    &\markovprob{\prog_1\fatsemi\prog_2}{\Diamond\opstate{\done}{\val'}} \\ 
    =& \sum_{\sigma\in\N^V} \semanticsaut{\Raut}(\sigma) \sum_{\sigma''\in\N^V}\mathrm{Pr}^{\markovchain{\Raut_\sigma}{\prog_1}}(\Diamond\opstate{\done}{\val''})\cdot \mathrm{Pr}^{\markovchain{\Raut_{\sigma''}}{\prog_2}}(\Diamond\opstate{\done}{\val'}) \\
    =& \sum_{\sigma\in\N^V} \semanticsaut{\Raut}(\sigma) \sum_{\sigma''\in\N^V}\semanticsaut{\semantics{\prog_1}(\Raut_\sigma)}(\sigma'')\cdot\semanticsaut{\semantics{\prog_2}(\Raut_{\sigma''})}(\sigma') \qquad\quad\text{(IH)} \\
    =& \sum_{\sigma\in\N^V}\sum_{\sigma''\in\N^V} \semanticsaut{\Raut}(\sigma)\cdot \semanticsaut{\semantics{\prog_1}(\Raut_\sigma)}(\sigma'')\cdot\semanticsaut{\semantics{\prog_2}(\Raut_{\sigma''})}(\sigma') \\
    =& \semanticsaut{\semantics{\prog_1\fatsemi\prog_2}(\Raut)}(\sigma')
\end{align*}
As for the probability of eventually reaching an error state, we have two options:
\begin{itemize}
    \item Reaching an error state from $\prog_1$ (e.g.\ $\prog_1 = \observe{\bfalse})$
    \item Executing $\prog_1$ with $\sigma$ (obtaining $\sigma'$) and reaching an error state from $\prog_2$ w.r.t.\ $\sigma'$
\end{itemize}
These two probabilities have to be calculated and their sum provides the overall probability of reaching the error state for some $\sigma\in\N^V$ (and the sum over all $\sigma$ provides the probability for all states).

\newcommand{\barprobmass}[1]{\ensuremath{
    \sum_{\bar{\sigma}\in\N^V} \semanticsaut{#1}(\bar{\sigma})
}}
\begin{align*}
    &\markovprob{\prog_1\fatsemi\prog_2}{\Diamond\errorstate} \\ 
    =& \sum_{\sigma\in\N^V} \semanticsaut{\Raut}(\sigma) \Bigl(\mathrm{Pr}^{\markovchain{\Raut_\sigma}{\prog_1}}(\Diamond\errorstate) + \\
    &\quad\sum_{\sigma'\in\N^V} \mathrm{Pr}^{\markovchain{\Raut_\sigma}{\prog_1}}(\Diamond\opstate{\done}{\sigma'})\cdot \mathrm{Pr}^{\markovchain{\Raut_{\sigma'}}{\prog_2}}(\Diamond\errorstate)\Bigr) \\
    =& \sum_{\sigma\in\N^V} \semanticsaut{\Raut}(\sigma) \Biggl(\barprobmass{\Raut_\sigma} - \barprobmass{\semantics{\prog_1}(\Raut_\sigma)} + \\
    &\quad\sum_{\sigma'\in\N^V} \Bigl(\semanticsaut{\semantics{\prog_1}(\Raut_\sigma)}(\sigma')\cdot \bigl(\barprobmass{\Raut_{\sigma'}} - \barprobmass{\semantics{\prog_2}(\Raut_{\sigma'})}   \bigr)\Bigr)\Biggr) \\
    =& \sum_{\sigma\in\N^V} \semanticsaut{\Raut}(\sigma) \Biggl(1 - \barprobmass{\semantics{\prog_1}(\Raut_\sigma)} + \\
    &\quad\sum_{\sigma'\in\N^V} \Bigl(\semanticsaut{\semantics{\prog_1}(\Raut_\sigma)}(\sigma')\cdot \bigl(1 - \barprobmass{\semantics{\prog_2}(\Raut_{\sigma'})}   \bigr)\Bigr)\Biggr) \\
    =& \sum_{\sigma\in\N^V}\Biggl( \semanticsaut{\Raut}(\sigma) - \sum_{\bar{\sigma}\in \N^V}\semanticsaut{\Raut}(\sigma)\cdot\semanticsaut{\semantics{\prog_1}(\Raut_\sigma)}(\bar{\sigma}) + \\
    &\quad \sum_{\sigma'\in\N^V} \Bigl( \semanticsaut{\Raut}(\sigma)\cdot \semanticsaut{\semantics{\prog_1}(\Raut_\sigma)}(\sigma') - \\
    &\quad\sum_{\bar{\sigma}\in\N^V}\semanticsaut{\Raut}(\sigma)\cdot \semanticsaut{\semantics{\prog_1}(\Raut_\sigma)}(\sigma')\cdot \semanticsaut{\semantics{\prog_2}(\Raut_{\sigma'})}(\bar\sigma) \Bigr) \Biggr)\\
    =& \sum_{\sigma\in\N^V} \semanticsaut{\Raut}(\sigma) - \sum_{\sigma\in\N^V}\Biggl( \sum_{\bar{\sigma}\in\N^V} \semanticsaut{\Raut}(\sigma) \cdot \semanticsaut{\semantics{\prog_1}(\Raut_\sigma)}(\bar{\sigma}) -\\
    &\quad \sum_{\sigma'\in\N^V}\Bigl(\semanticsaut{\Raut}(\sigma)\cdot \semanticsaut{\semantics{\prog_1}(\Raut_\sigma)}(\sigma') - \\
    &\quad\sum_{\bar{\sigma}\in\N^V}\semanticsaut{\Raut}(\sigma)\cdot \semanticsaut{\semantics{\prog_1}(\Raut_\sigma)}(\sigma')\cdot \semanticsaut{\semantics{\prog_2}(\Raut_{\sigma'})}(\bar\sigma)  \Bigr) \Biggr) \\
    =& \probmass{\Raut} - \probmass{\semantics{\prog_2}(\semantics{\prog_1}(\Raut))} \\
    =& \probmass{\Raut} - \probmass{\semantics{\prog_1\fatsemi\prog_2}(\Raut)}
\end{align*}
\end{proof}
With this lemma, we can now show the overall operational equivalence for all \redip program (excluding the \texttt{iid}-statement).
\opequiv*

\begin{proof}

\begin{align*}
&\markovprob{\prog}{\Diamond\opstate{\done}{\val}~\vert~\neg\Diamond\errorstate} \\ 
=&\quad \frac{\markovprob{\prog}{\Diamond\opstate{\done}{\val} \land \neg \Diamond\errorstate}}{\markovprob{\prog}{\Diamond\neg\errorstate}} \\
=&\quad \frac{\markovprob{\prog}{\Diamond\opstate{\done}{\val}}}{\markovprob{\prog}{\Diamond\neg\errorstate}}  \\
=&\quad \frac{\markovprob{\prog}{\Diamond\opstate{\done}{\val}}}{\probmass{\Raut}-\markovprob{\prog}{\Diamond\errorstate}} \\
=&\quad \frac{\semantics{\prog}(\Raut)(\val)}{\probmass{\Raut} - (\probmass{\Raut} - \semanticsaut{\semantics{\prog}(\Raut)})} \quad \text{(\Cref{thm:non-norm-equiv})}\\
=&\quad \frac{\semantics{\prog}(\Raut)(\val)}{\semanticsaut{\semantics{\prog}(\Raut)}} \\
=&\quad \normalize{\semantics{\prog}(\Raut)}(\val)
\end{align*}
One can remove the conjunction in step two since terminating implies not reaching an error state. Note that if $\markovprob{\prog}{\Diamond\errorstate} = \probmass{\Raut}$, we have $\probmass{\semantics{\prog}(\Raut)} = 0$, and thus both are undefined.
\end{proof}

\end{document}